\documentclass[11pt,twoside]{article}
 \makeatletter
\newcommand{\specificthanks}[1]{\@fnsymbol{#1}}
\makeatother

\usepackage{latexsym}
\usepackage{xcolor}
\usepackage{amssymb,amsbsy,amsmath, amsfonts,amssymb,amscd}
\usepackage{graphicx}
\usepackage{caption}
\usepackage{comment}
\usepackage{hyperref}
\usepackage{cancel}
\usepackage{mathenv}
\newcommand{\commentout}[1]{}
\newcommand{\R}{\mathbb{R}}

\newcommand {\Chi} {{\bf \raise 2pt \hbox{$\chi$}} }
\newcommand{\beq}{\begin{equation}}
\newcommand{\eeq}{\end{equation}}
\newcommand{\bea} {\begin{array}{rl}}
\newcommand{\eea} {\end{array}}
\newcommand{\bepa}{\left\{ \begin{array}{l}}
\newcommand{\eepa} {\end{array}\right.}
\newtheorem{theorem}{Theorem}[section]
\newtheorem{lemma}[theorem]{Lemma}

\newtheorem{remark}[theorem]{Remark}
\newtheorem{proposition}[theorem]{Proposition}

\newcommand{\qed}{{ \hfill
                       {\unskip\kern 6pt\penalty 500 \raise -2pt\hbox{\vrule\vbox to 6pt{\hrule width 6pt
                       \vfill\hrule}\vrule} \par}   }}
 
\title{Toward an integrated workforce planning framework using structured equations}
  
\author{ \scriptsize{Marie Doumic}\thanks{Inria de Paris, EPC Mamba,  UPMC et CNRS, F75005 Paris, France}
\and{\scriptsize{Beno\^ \i t Perthame}\thanks{Sorbonne Universit\'es, UPMC Univ Paris 06, Laboratoire Jacques-Louis Lions  UMR CNRS 7598,  Inria, F75005 Paris, France} } 
  \and{\scriptsize{Edouard Ribes}\thanks{Email: edouard.augustin.ribes@gmail.com} }
  \and{\scriptsize{Delphine Salort}\thanks{Sorbonne Universit\'es, UPMC, Laboratoire de Biologie Computationnelle et Quantitative UMR CNRS 7238, F75005 Paris, France}}
     \and{\scriptsize{Nathan Toubiana}\thanks{nathan.toubiana@polytechnique.edu} }}

\date{\today}
\begin{document}
\maketitle
\pagestyle{plain}
\pagenumbering{arabic}
\begin{abstract} 
Strategic Workforce Planning is a company process providing best in class, economically sound, workforce management policies and goals. Despite the abundance of literature on the subject, this is a notorious challenge in terms of implementation. Reasons span from the youth of the field itself to broader data integration concerns that arise from gathering information from financial, human resource and business excellence systems.\\ 
\hspace*{0.2cm} This paper aims at setting the first stones to a simple yet robust quantitative framework for Strategic Workforce Planning exercises. First  a method based on structured equations is detailed. It is then used to answer two main workforce related questions: how to optimally hire to keep labor costs flat? How to build an experience constrained workforce at a minimal cost? 
\end{abstract} \vskip .7cm

\noindent{\makebox[1in]\hrulefill}\newline

\section{Introduction}
\label{sec:intro}
  
  Strategic Workforce Planning (SWP) examines the gap between staff availabilities (internal and external to the organization) and staffing requirements (to perform tasks in the organization) over time, and prescribes courses of action to narrow such a gap (\cite{doi:10.1108/01437729610110602}). Multiple methodologies exist to sustain it. They all revolve around 5 milestones (\cite{bcg,young2006strategic}): after a first baselining of the population, demographic forecasts are drafted in order to assess the potential evolution of a company's headcount. Then business needs, both in terms of headcount and competencies, are gathered to perform a gap analysis between a company's desired future state and its natural evolution. Finally solutions to bridge the gaps are proposed, agreed upon and implemented.\\ 
\hspace*{0.2cm} If the process in itself seems simple and if many research studies are focused on the topic of strategic workforce planning (see state of the art), SWP is something most companies struggle to implement (\cite{guthridge2008making}). According to the Corporate Executive Board (CEB) latest benchmarks (\cite{talentneuron,changinghigg}), only 10\% of companies really succeed in aligning their workforce plans to meet strategic objectives.  Among the surveyed firms, 70\% failed at drafting a workforce plan and 84\% of them are not confident in their use of labor market trends. The same study stated that 65\% of the respondents felt a disconnection between the business needs and standard Human Resources processes such as recruitment. Therefore, there is a need to jump from methodological milestones to analytics in order to standardize and industrialize the technical aspects of SWP.

\paragraph{State of the art.}

{SWP is a research field which emerged in the 70s, see for instance the seminal books~\cite{bartholomew1976manpower,vajda1978mathematics}. Stochastic formalisms are prominent in the field, including Markov chains and stochastic linear programming, game theory,  convex approximation etc. - see e.g.~\cite{Song200829,georgiou2002modelling,mcclean1997non}  or yet~\cite{DeBruecker20151} for a recent review. 
}


{ In this corpus, s}ome studies aim at determining an optimal hiring policy, { which is also a key motivation of our approach.} For instance, E. G. Anderson found the optimal policy  by searching the best ratio between apprentices and experienced employees, in a growth context, with a model based on experience and productivity which suggests to strike the happy medium between too many apprentices (that have to be trained by older employees) and too many experienced employees (that are more expensive in the company's point of view) \cite{anderson2001managing}. Other studies also proposed to optimize the required number of staff with a stochastic model \cite{bartholomew1977maintaining}.

{ Rare studies use partial differential equations (PDE) in the framework of population dynamics, see e.g. \cite{gerchak1990manpower}. They seem to be very marginal in the field untill now. Hence, our study aims at providing the first building blocks to a comprehensive approach using these so-called \emph{structured population equations}.}

{ Though only rarely applied to SWP, deterministic}  population dynamics has been an extensive research topic, which fields of application are very broad, especially in biology, where partial differential equations (PDE) are frequently used to model real life processes in ecology, immunology, epidemiology (\cite{edelstein1988mathematical,perthame2006transport, thieme2003mathematics})... A subject that started with Malthusian considerations has now evolved into advanced multidimensional and nonlinear frameworks. Among structured population models, the age-structured, also called "renewal" or McKendrick-Von Foerster equation \cite{kermack1927contribution, keyfitz1997mckendrick}, is one of the most widely used and studied equation, under linear or nonlinear forms, and with variants  used in many fields, from the neuroscience to cancer modeling. 

\paragraph{Goals and motivations.}

Companies' Financial Information Systems (IS) and/or Human Resources Information Systems (HRIS) collect both labor costs and demographic data as part of their standard processes. In section \ref{section2}, the proposal developed in this paper revolves around creating an actionable quantitative framework based upon those data. This enables a workforce evolution forecast and provides a better understanding of the dynamics at stake to manage a company workforce. In section \ref{section3}, the explanatory power of this framework is stressed by its results on standard workforce management policies. It is shown that moving from a workforce management by operating expenses toward an optimization of the overall workforce experience is economically sound. Empirical evidence is provided.
 \\
  \\
\hspace*{0.2cm}This article is organized as follows. In section \ref{section2}, we build a preliminary framework with which we determine the workforce evolution and convergence towards a stable age structure. We show that there can be many short term headcount fluctuations, and studying the long term behavior may not be appropriate, due to an exceedingly long time scale. We therefore build another framework in section \ref{section3} for which the hire rate structure is driven by an economic constraint: the labor cost. We first determine the workforce evolution,{ we then optimize the company's expenses with maintained experience, which leads us to an optimal demographic structure and an associated hiring policy.}

\section{{Analyze} workforce evolution in a demographic framework}
\label{section2}
 
{SWP is usually a long term analysis. Hence, assessing the stability of a company workforce is of key importance. One's workforce usually evolves according to its demographics characteristics (age, tenure, gender ....). Two main movements rules this evolution: attrition and hiring. Attrition accounts for workers leaving the company. Hiring is endogenous (depending on firm activity) while attrition is exogenous. Attrition is driven by three factors: market labor demand, company termination policies and retirement. In this specific case, company induced terminations are not allowed and employees only leave the company according to their own wish. Retirement is taking into account by introducing a retirement age $z_{\rm max}$ after which the worker leaves the active workforce. }

{ To model this evolution, we consider the population of workers of age $z$ at time $t$. We denote it $\rho(t,z),$ with $z\in [z_{\rm min},z_{\rm max}]$, $z_{\rm min}$ being the youngest hiring age and $z_{\rm max}$ the retirement age. 
In this first study, we also assume that attrition
 is purely exogeneous with a rate depending only on the \emph{age} of the workers, and that the hiring policy  determines a hiring age distribution, denoted by $\gamma (z),$ and a certain hiring rate depending only on the total population $P_t=\int_{z_{\rm min}}^{z_{\rm max}} \rho(t,z)dz .$ One can notice that such assumptions are quite strong: in many cases, hiring should depend also on other so-called \emph{structuring} variables, and not only age - for instance, experience, skills, gender, etc. Similarly, attrition could depend on the same kind of factors. However, our simplifying assumptions allow us to build a self-consistant  example, already able to give useful insights in the evolution of the population.
We thus write the following age-structured equation satisfied by $\rho(t,z)$:}
 $$
 \overbrace{\frac{\partial\rho}{\partial t}(t,z)+\frac{\partial\rho}{\partial z}(t,z)}^{Workforce \ evolution}=-\overbrace{\mu (z)\rho(t,z)}^{Attrition}+\overbrace{h(P_t)P_t\gamma (z)}^{Hiring},\qquad z_{\rm min}<z<z_{\rm max},
 $$
 where $\mu (z)$ is the attrition rate, and  $h(P_t) P_t \gamma (z)$ is the  population hired at size $z$. We assume that $\mu$ and $\gamma$ are independent of time because the current framework is built for businesses with long product and research cycles (typically 5 to 10 years), which translates into a relatively stable global labor competition and experience needs. The coefficient $h(P_t)P_t$ represents the hiring rate for the population in scope. 
{ We choose to write this rate as the product of $P_t$ modulated by a function $h(P_t)$ because in a "reasonable" population range, it is natural to build a model where the number of hired employees is proportional to the total population. This would correspond to $h$ independent of $P_t$. However, if $h$ is constant for all $P_t,$ the model  becomes linear, leading to an exponential growth or decay of the population, { see e.g.~\cite{perthame2006transport}}. Hence, we need to choose a saturation $h$ such that the population cannot explode - see below.}
 
\subsection{Identifying the hiring rate structure}
  
Consequently, we consider here that the hiring profile $\gamma$ has been defined and propose another hiring rate based only on the total headcount $P_t$. We study its ability to stabilize the workforce population towards an age profile $P_{\rm eq}$. {As mentioned in E. Gurtin work (\cite{gurtin1974non}), for the Malthusian law, the birth and death moduli are independent of the population, which does not reflect a realistic workforce behavior}. Using a standard formulation in population evolution, we choose the saturation rate under the form $h(P_t)= \frac{1}{1+\alpha P_t^2}$. Therefore, the temporal evolution of the headcount density is driven by this equation:
\begin{equation}
\left\{
\begin{array}{l}
 \frac{\partial\rho}{\partial t}(t,z)+\frac{\partial\rho}{\partial z}(t,z)=-\mu (z)\rho(t,z)+\frac{P_t}{1+\alpha P_t^2}\gamma (z), \qquad z_{\rm min}<z<z_{\rm max},\\
  \rho(t,z_{\rm min})=0,\\
  \rho(0,z)=\rho^0(z)\geq 0.
\end{array}
\right.
\label{eq1}
\end{equation}
 \hspace*{0.2cm}The parameter $\alpha$ is a pressure population constant representing the budget constraint ($\alpha > 0$). Indeed, as the parameter $\alpha$ is positive, the hiring rate increases with the population for small populations, and decreases from a certain population threshold. So workforce cannot grow exponentially, which reflects the fact that companies cannot hire indefinitely. { Notice however that the choice of the shape of $h(P_t)$ is somewhat arbitrary: similar reasoning could be done with any decreasing function, departing from a sufficiently high value at $0$ and vanishing at infinity.} 
 \\ 
 \hspace*{0.2cm}The hiring age distribution $\gamma (z)$ is set to its historical value. Under this formalism, stability can be reached. The convergence (see appendix \ref{appen1}) is achieved exponentially fast. 
In order to ensure a non null steady state, we show that the following condition is required:
\begin{equation}\label{toto}
\beta:=\int_{z_{\rm min}}^{z_{\rm max}}{\left( \int_{z_{\rm min}}^{z}{\gamma(y)e^{-\left(M(z)-M(y) \right) }dy} \right)dz} > 1,
\end{equation} 
where $M$ is an antiderivative of $\mu$. This may be interpreted as the fact that the hiring rate must be sufficiently high to counterbalance those leaving the firm.

      \subsection{How to action the framework}\label{numex1}
      
    In the case of a non null equilibrium, the hiring rate structure and the steady state $P_{\rm eq}$ of the workforce are closely connected. Indeed, {considering the equilibrium equation,} we show in the appendix \eqref{ap} that: 
      $$\alpha=\frac{\int_{z_{\rm min}}^{z_{\rm max}}{\left(\int_{z_{\rm min}}^{z}{\gamma(y)e^{-(M(z)-M(y))}dy}\right)dz}-1}{P_{\rm eq}^2},$$which leads us to the condition \eqref{toto} as $\alpha > 0$. 
    \\
\hspace*{0.2cm}{Consider the case of a company who is interested in building a stable workforce under unlimited contract of size $P_0$ while adjusting for labor demand through temporay contract. It can be assumed that its overall workforce is not likely to change over the long term ($P_0=P_{\rm eq}$). The company decision could, for instance, be motivated by long training times required to develop expertise in the workforce (ex: research in medical fields). According to the previous formalism, the hiring rate is hence fixed. In the next subsection, we analyze the short term workforce evolution according to the current workforce demographic structure.}
 
    \paragraph{Examples: necessity to adjust workforce management practices to reach stability.}
    
We choose to display the workforce analysis for two cases. For both examples, we show the initial workforce structure, the attrition and the hired population distribution, and we then display the associated workforce evolution. We assume $P_{\rm eq}=P_0=1000$ for both cases. The first example is taken in a fictional business unit A (BU A). In this example, the turnover rate is very low, and employees usually wait until retirement to leave the firm. The second example is taken in another fictional business unit B (BU B). In this example, employees are mainly young, and tend to leave the firm quickly. This is typically the case for sectors in which there are specific labor policies revolving around fixed term contracts and extreme labor demand. The numerical method is described in the appendix \eqref{appen}.
  \\
    \\
       \includegraphics[scale=0.21]{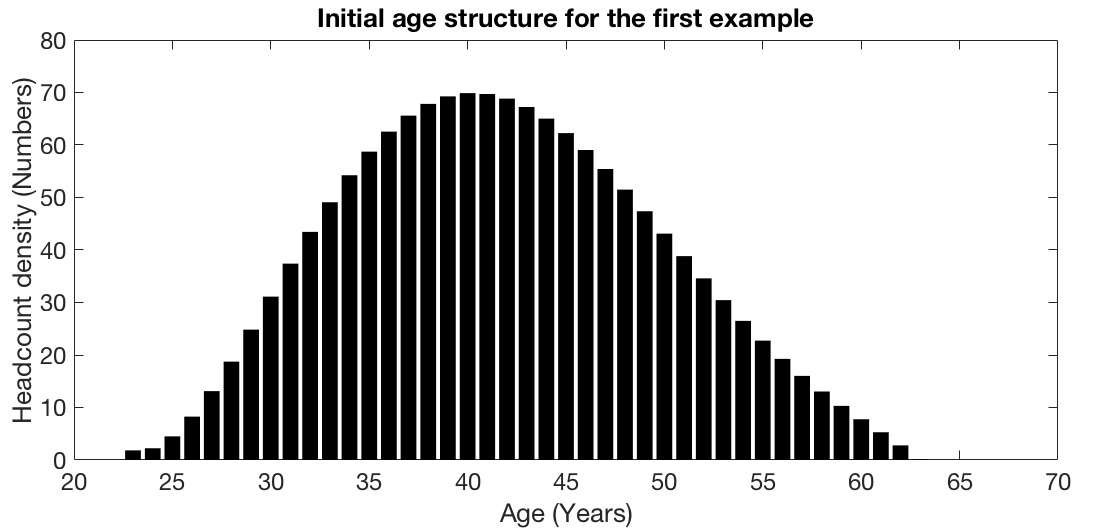}
\includegraphics[scale=0.21]{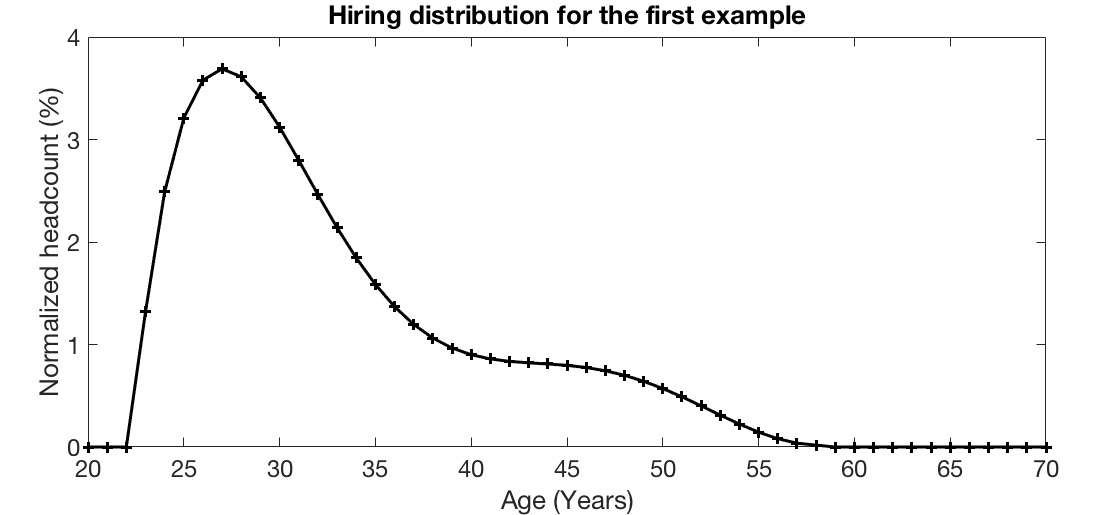}
  \\
\hspace*{4cm} \includegraphics[scale=0.23]{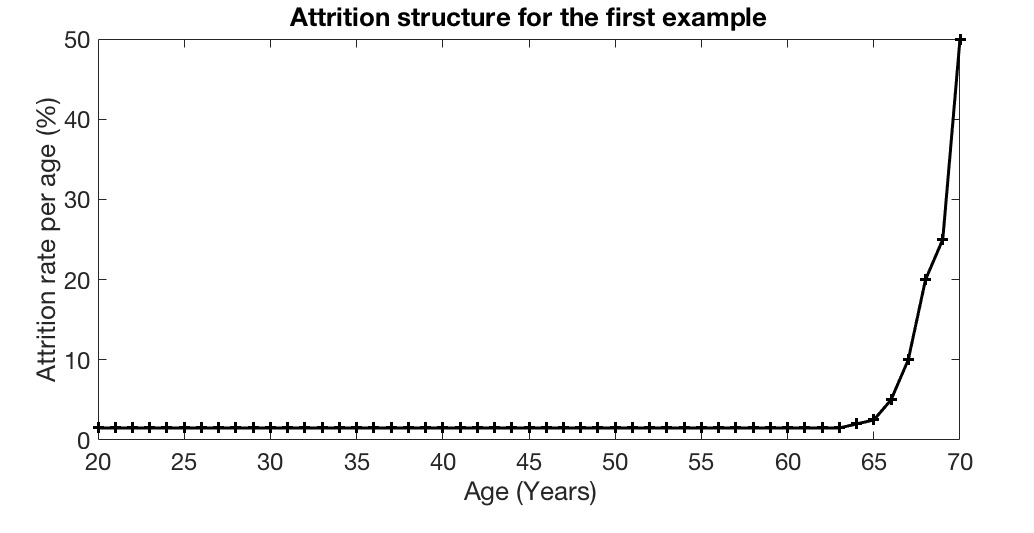}
  \\
    \\[-8mm]
\begin{scriptsize}
\hspace*{0.2cm}{FIG. 1. \textit{Initial age structure, historical hired population distribution (normalized), and historical attrition rate (for $z_{\rm min}=20 $ years and $z_{\rm max}=70 $ years) for the BU A.   \\ \\}}  
\end{scriptsize}
\includegraphics[scale=0.22]{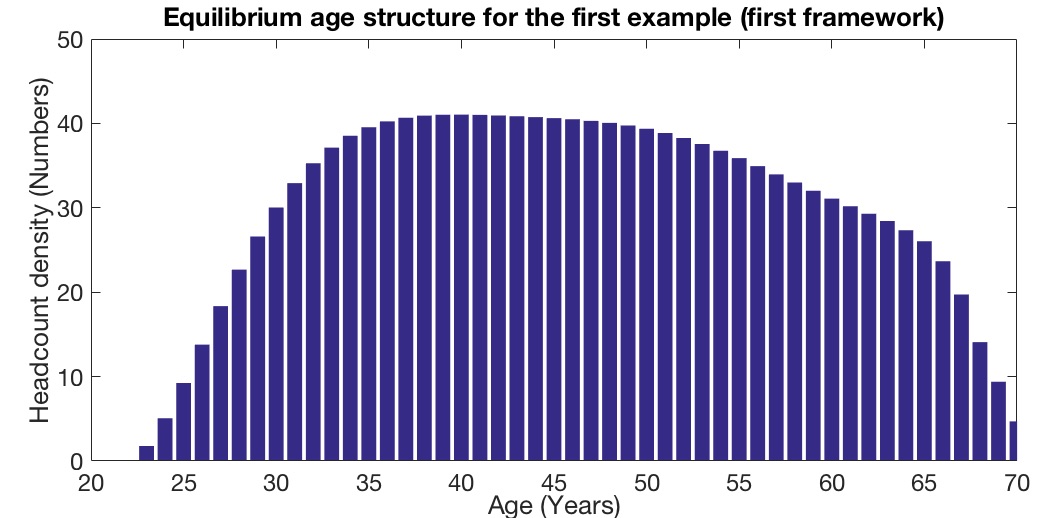}
 \includegraphics[scale=0.225]{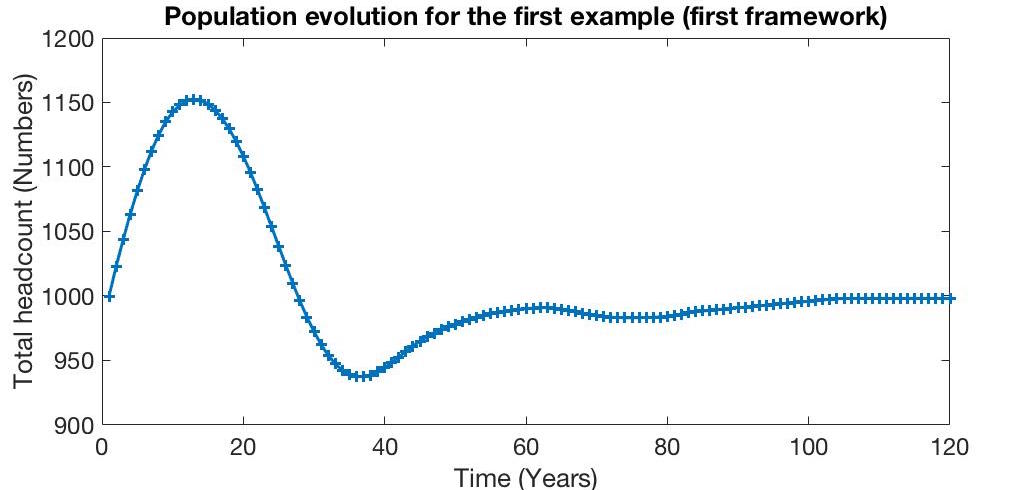}
   \\
   \\[-6mm]
\begin{scriptsize}
\hspace*{0.2cm} {FIG. 2. \textit{Equilibrium age structure and headcount temporal evolution for the BU A, for the discretization $\delta t=\delta z=1 $ year, and for $P_{\rm eq}=P_0=1000$.}}
\end{scriptsize}
  \\   \\
\includegraphics[scale=0.21]{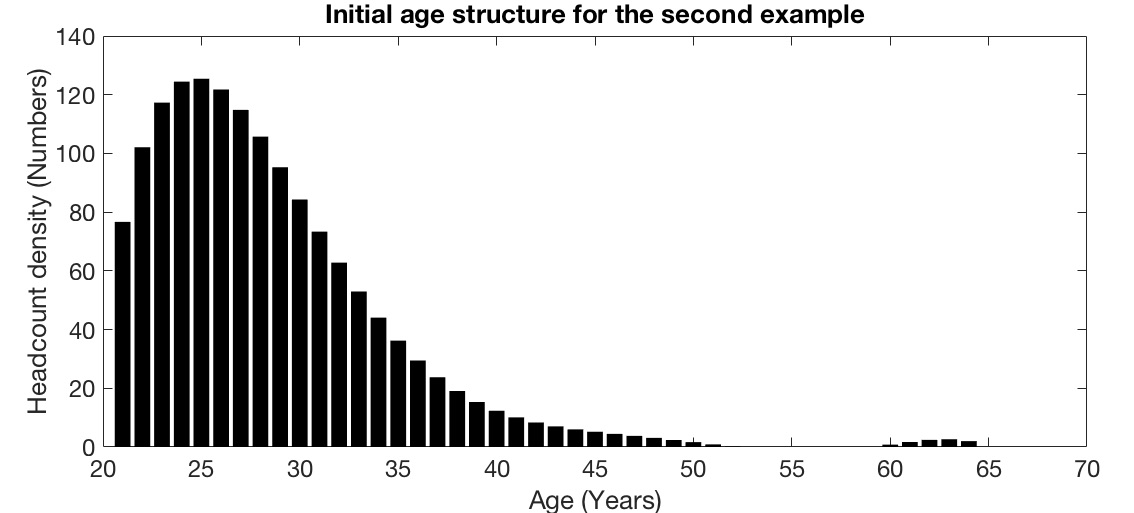}
\includegraphics[scale=0.21]{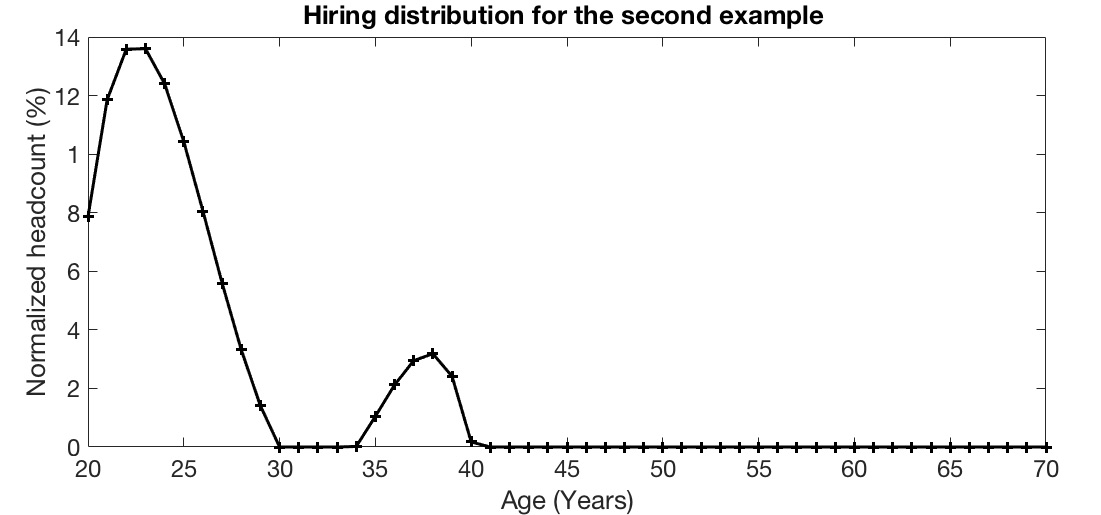}
  \\
\hspace*{4cm}\includegraphics[scale=0.23]{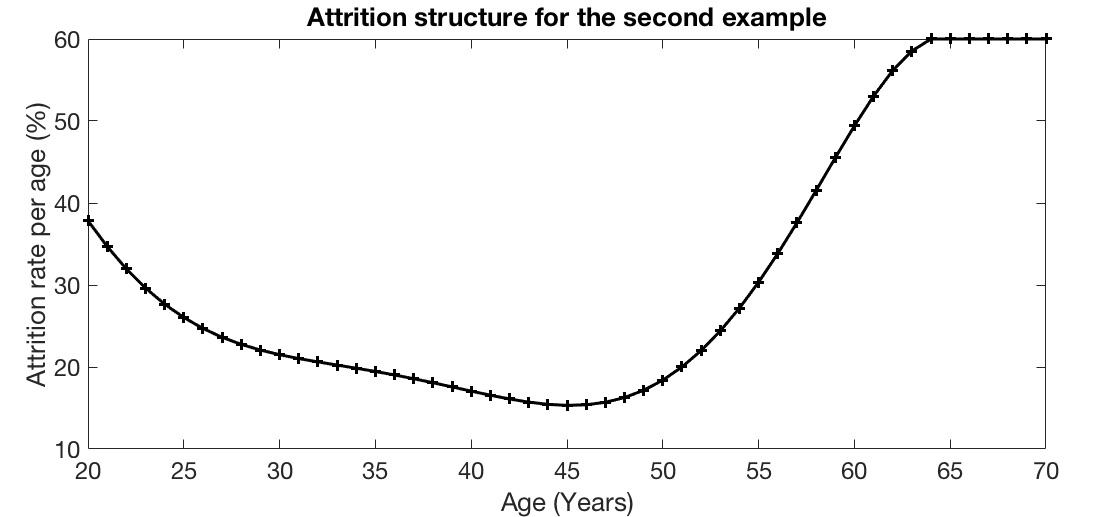}
 \\
   \\[-6mm]
\begin{scriptsize}
\hspace*{0.2cm} {FIG. 3. \textit{Initial age structure, historical hired population distribution (normalized), and historical attrition rate (for $z_{\rm min}=20 $ years and $z_{\rm max}=70 $ years) for the BU B.}}
\end{scriptsize}
 \\
  \\
\hspace*{0.2cm}\includegraphics[scale=0.22]{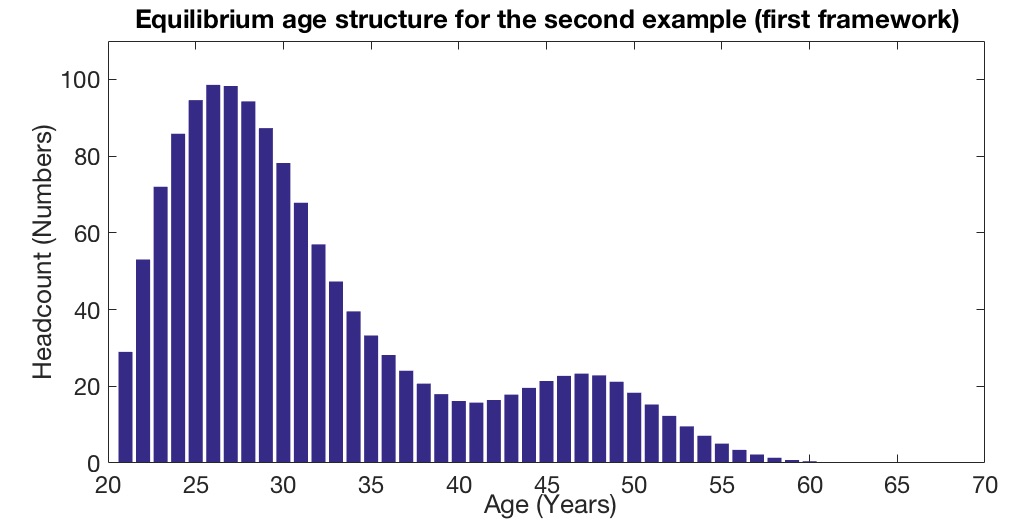}
\includegraphics[scale=0.22]{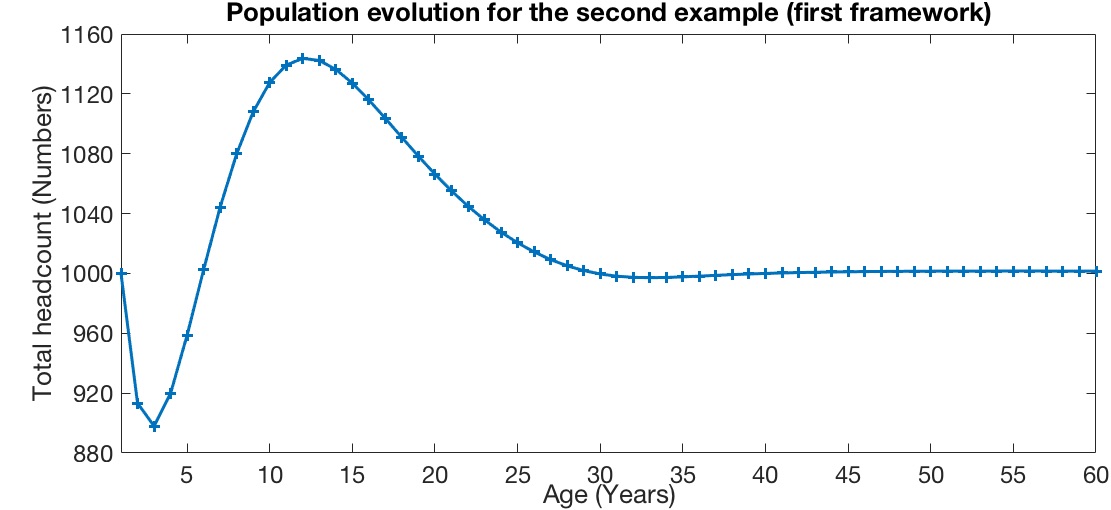}
 \\
    \\[-6mm]
\begin{scriptsize}
\hspace*{0.2cm} {FIG. 4. \textit{Equilibrium age structure and headcount temporal evolution for the BU B, for the discretization $\delta t=\delta z=1 $ year, and for $P_{\rm eq}=P_0=1000$. \\  \\}}
\end{scriptsize}
For the BU A, we can see that the initial average age is approximately 45 years. Furthermore, employees are mostly hired when they are young, and the maximum attrition rate is at retirement (Figure 1). The final average age of the employees is also 45 years, so the overall population did not age. This is due to the high hiring rate for young employees and the very low attrition rate for all employees until retirement. This also results in a flattening of the age structure. Plus, we note that the equilibrium is reached within approximately 80 years, and there are substantial headcount fluctuations in-between (Figure 2). 
\\  
\hspace*{0.2cm}For the BU B, we can see that the initial average is approximately 27 years. Furthermore, employees are mostly hired when they are young, and the maximum attrition rate is both for the youngest (fixed term contracts) and oldest (retirement) employees (Figure 3). The final average age of the employees is approximately 30 years, 3 years older than the initial average age, which is due to the hiring profile and the attrition rate: young and old employees tend to leave {early} the company, whereas average-aged employees stay (and age) in the company. Plus, we note that the equilibrium is reached within 30 years, and there are substantial fluctuations in-between (Figure 4).
 \\   
\hspace*{0.2cm}As a whole, we find that the equilibrium state is reached very slowly (80 and 30 years), and the fluctuations that we first thought to be short term may not be as short as expected. Indeed, fluctuations can extend up to 60 years, which is higher than an employee's lifetime in the company. 
  \\   
\hspace*{0.2cm}Although determining the steady state seems conceptually appealing, it may not be a relevant option, since the equilibrium will not be reached in a company's activity time scale range. In the next section, we review and modify the hiring rate structure, according to a reasonable economic constraint. The functional $a(t)=\frac{P_t}{1+\alpha P_t^2}$ has been designed empirically to answer good qualitative properties to the solution, the parameter $\alpha$ being determined by the target equilibrium $P_{\rm eq}$, which happens to be achieved too late to be sound. Plus, each employee does not necessarily have the same impact on the hiring policy of the firm, and this first hiring rate structure does not translate this idea.

\section {Design of economically sustainable management policies}
\label{section3}
      
    {We now consider another hiring policy based on budget considerations. We assume that employees have a certain cost depending on their age. In the first subsection, the workforce evolution will be analyzed with a total budget constraint. In the second subsection, an ideal hiring policy will be investigated in order to minimize the cost while keeping a fixed total experience.}

\subsection{Management policy 1: operational expenditure (opex) adjustments}\label{sub31}
        
      {As a first step, we assume that the total annual budget (which is assimilated to the sum of the annual salaries) remains constant at all times. This drives the hiring policy through the modulation of the hiring rate.}
        This translates into the following age-structured representation:
        \begin{equation}\label{eq3}
        \left\{
\begin{array}{l}
  \frac{\partial\rho}{\partial t}(t,z)+\frac{\partial\rho}{\partial z}(t,z)=-\mu (z)\rho(t,z)+h([\rho])\gamma (z), \qquad z_{\rm min}<z<z_{\rm max},  \\
  \rho(t,z_{\rm min})=0,\\
  \rho(0,z)=\rho^0(z)\geq 0, 
\end{array}
\right.
        \end{equation}
       where $h([\rho])$ depends on the labor cost contraint and does not depend on the age $z$. 
         \paragraph{Quantitative framework.}\label{numex2}
         
 To find the hiring rate structure $h([\rho])$, we assume here that the hiring profile $\gamma$ is given and that the total budget $\int_{z_{\rm min}}^{z_{\rm max}}{\rho (z,t)\omega (z)dz}$ {is} not time-dependent. The  the cost per employee $w(z)$ of age $z$ is given as well. By definition, this makes the equation conservative. {Indeed,} we have 
 $$
 \omega (z)\frac{\partial\rho}{\partial t}(t,z)+\omega(z)\frac{\partial\rho}{\partial z}(t,z)=-\omega(z)\mu (z)\rho(t,z)+\omega(z)h([\rho])\gamma (z),  
 $$
  and 
  $$
  \int_{z_{\rm min}}^{z_{\rm max}}{\omega (z)\frac{\partial\rho}{\partial t}(t,z)dz}=\int_{z_{\rm min}}^{z_{\rm max}}{\frac{\partial\rho\omega}{\partial t}(t,z)dz}=0 ,
  $$
   so  
   $$\underbrace{\int_{z_{\rm min}}^{z_{\rm max}}{\omega(z)\frac{\partial\rho}{\partial z}(t,z)dz}}_{\omega (z_{\rm max})\rho (t,z_{\rm max})-\int_{z_{\rm min}}^{z_{\rm max}}{\rho(z)\frac{\partial\omega}{\partial z}(t,z)dz}}=-\int_{z_{\rm min}}^{z_{\rm max}}{\omega(z)\mu (z)\rho(t,z)dz}+\int_{z_{\rm min}}^{z_{\rm max}}{\omega(z)h([\rho])\gamma (z)dz}, 
   $$
    and thus we obtain the following formula for the hiring rate  
    $$
    h([\rho])=\frac{\overbrace{\int_{z_{\rm min}}^{z_{\rm max}}{\omega(z)\mu (z)\rho(t,z)dz}}^{Attrition}+\overbrace{\omega (z_{\rm max})\rho (t,z_{\rm max})}^{Retirement}-\overbrace{\int_{z_{\rm min}}^{z_{\rm max}}{\rho(t,z)\frac{\partial\omega}{\partial z}(z)dz}}^{Cost\, of\, aging}}{\int_{z_{\rm min}}^{z_{\rm max}}{\omega(z)\gamma (z)dz}}. 
    $$
     So $h$ {yields} a linear form, which is easy to interpret: 
         \begin{itemize}
          \item The first term represents the budget available {resulting} from employees attrition.
           \item The second term represents the budget available {because of retirement at} age $z_{\rm max}$. 
            \item The last term is the cost of aging, which tracks the drift in wages due to seniority and promotions.
          \end{itemize}
          Under this formalism, and with the assumption $ \mu \omega \geq \omega^{'} $ (which may be interpreted as a positive balance between the budget earned with the attrition and to the cost of aging), stability can be reached. The convergence is shown in the appendix \eqref{ap2}.
          
                 \paragraph{Examples.}
                 
      Now, we can analyze the workforce evolution for this framework, with the same two examples of the BUs A and B. The historical values (initial age structure, attrition rate and hiring distribution) are the same as before. The numerical method is described in the appendix \eqref{ap3}.
   \\
    \\
      \hspace*{2.5cm} \includegraphics[scale=0.3]{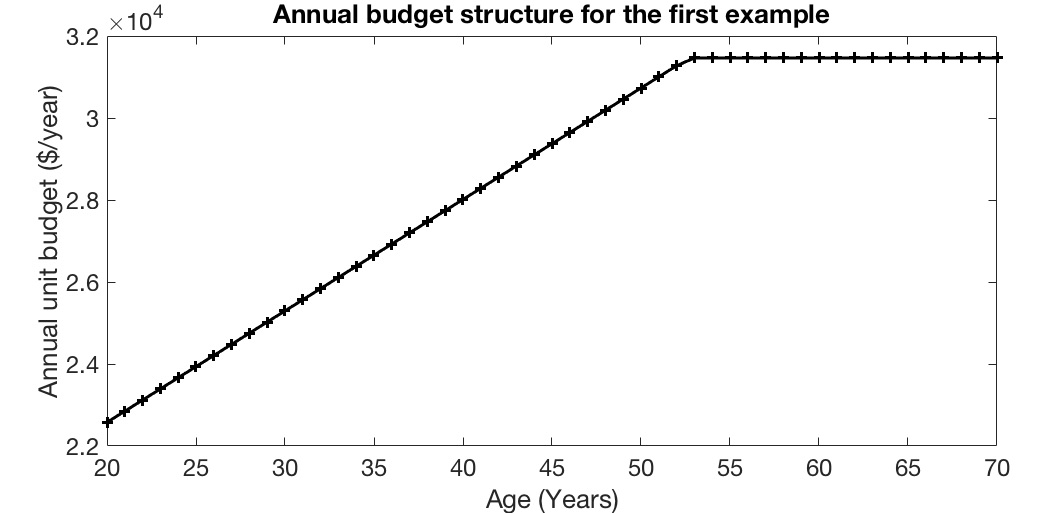}
\begin{scriptsize}
 \\
    \\[-6mm]
  \hspace*{0.2cm} {FIG. 5. \textit{Budget structure $\omega(z)$ of the employees of the BU A. \\ }}
\end{scriptsize}
\includegraphics[scale=0.22]{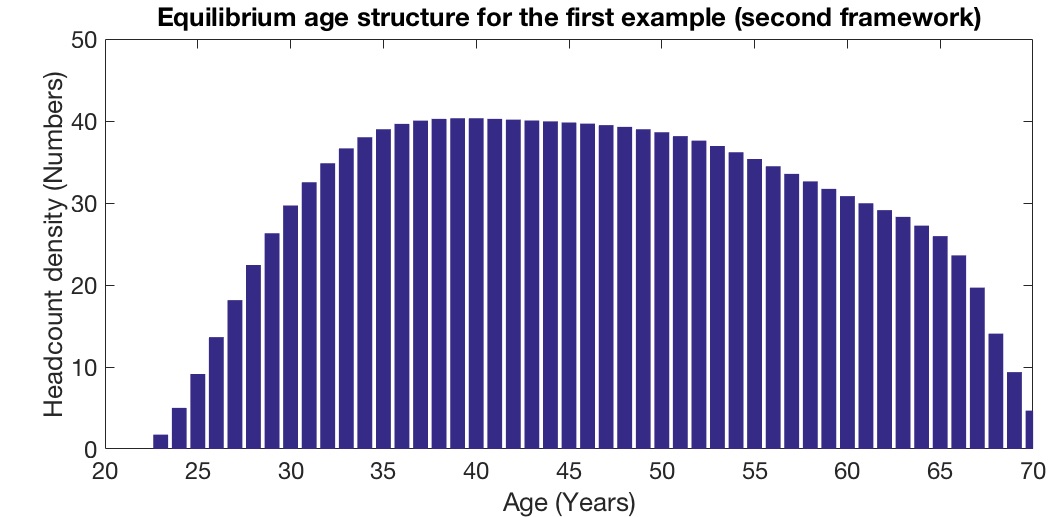}
\includegraphics[scale=0.24]{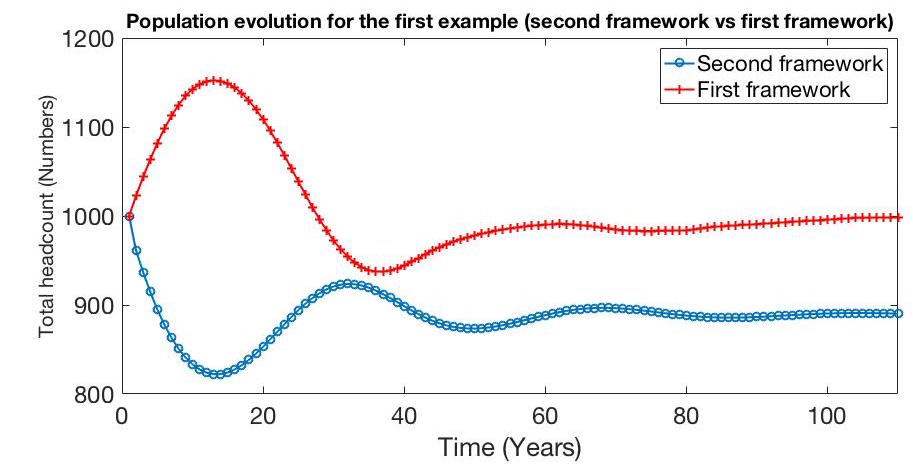}
   \\
      \\[-6mm]
\begin{scriptsize}
\hspace*{0.2cm} { FIG. 6. \textit{Equilibrium age structure and headcount temporal evolution for the BU A, for the discretization $\delta t=0.5$ year and $\delta z=1 $ year.  \\   \\ }}
\end{scriptsize}
   \hspace*{3cm}     \includegraphics[scale=0.3]{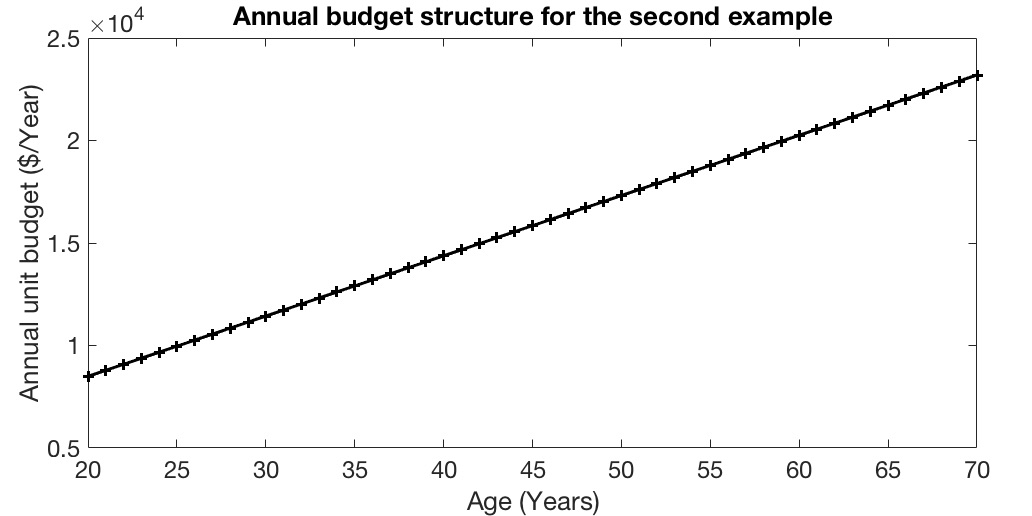}
\begin{scriptsize}
 \\ 
    \\[-6mm]
 \hspace*{0.4cm} {FIG. 7. \textit{Budget structure $\omega(z)$ of the BU B. }}
\end{scriptsize}
 \\
  \\
 \hspace*{0.2cm}\includegraphics[scale=0.215]{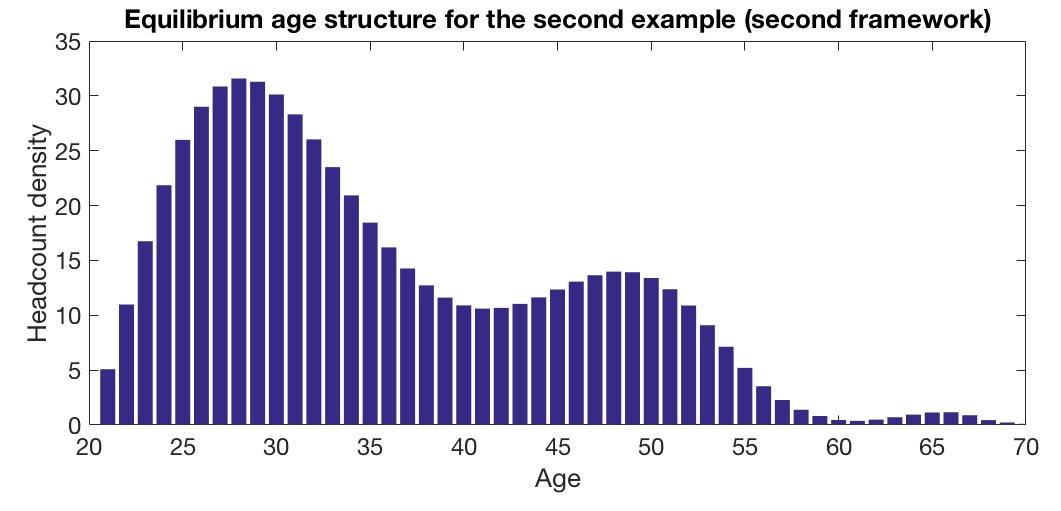}
\includegraphics[scale=0.22]{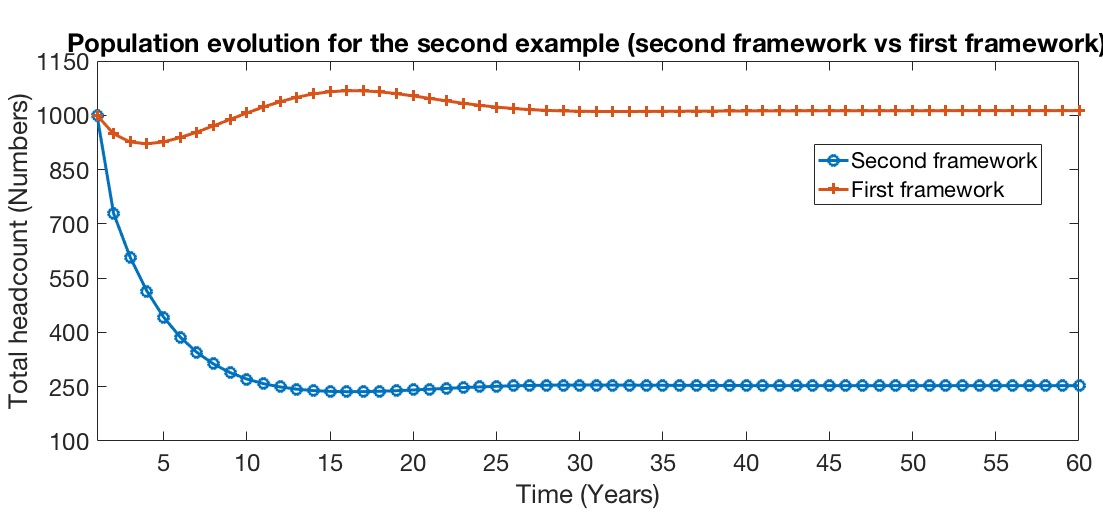}
 \\
  \\[-6mm]
\begin{scriptsize}
\hspace*{0.2cm} { FIG. 8. \textit{Equilibrium age structure and headcount temporal evolution for the BU B, for the discretization $\delta t=0.5$ year and $\delta z=1 $ year.   \\  \\}}
\end{scriptsize}
The budget structure of the employees of BU A (Figure 5) is linear and increases with age, until a certain age (approximately 55). {It then stabilizes} reflecting the fact that the maximum experience for this type of employee is reached at approximately 55 years. We can see that the final age structure is very similar to the one of the first framework (Figures 2 and 6). {Note that, as in the previous framework, the hiring policy favoring young employees and the very low attrition rate} result in a flattening of the age structure. However, the final headcount is 10\% lower (approximately 900 instead of 1000) {than its current baseline}. The equilibrium is reached within 90 years (Figure 6). 
\\
\hspace*{0.2cm} On the other hand, the budget structure of the employees of BU B is fully linear (Figure 7). We can see that the final age structure is very similar to the first framework (Figures 4 and 8). Just as in the previous framework, young and old employees tend to leave quickly the company, whereas average-aged employees stay in the company. However, the equilibrium headcount is 75\% lower (250 instead of 1000), which is due to the flat total budget constraint while having an aging population. The equilibrium is reached within 30 years (Figure 8).
  \\
\hspace*{0.2cm} For both the first and second examples, the equilibrium age structures are very similar for the two frameworks. However, the equilibrium headcount is different (in these cases lower), because we did not fix $P_{\rm eq}=P_0$ for the second framework. Plus, {the evolution to} the stable state depends on the framework. By simulating several cases with diverse assumptions, we observe empirically that there seems to be less oscillations for the first one, and the time to reach the equilibrium state is similar for both frameworks.
 \\
\hspace*{0.2cm}Even though the two frameworks are similar (in terms of fluctuation and stability), the second one may be more adapted to the SWP analysis. Indeed, this framework makes more economical sense and {add a differentiation layer among employees beyond their loyalty to the company (illustrated through the attrition rate $\mu(a)$)}.

  \subsection{Management policy 2: invest in knowledge}\label{sub32}
  
  Until now, we have kept the hired population distribution constant equal to its historical values. Though this is convenient to analyze the natural workforce evolution, identifying the optimal hiring policies is of key importance regarding the business needs assessment of a given company. This is why the hired population distribution $\gamma (z)$ is not fixed anymore, and neither is the total budget.
  
\paragraph{Identification of the optimal hiring policy.}
   
    We now minimize the global labor cost with given total knowledge, and hence find an optimal age structure and an optimal hiring policy. We consider the case of knowledge workers, in fields for which specific knowledge is required (for instance: experts from the medical field). Knowledge is the sum of aggregated experience and is age dependent. In this case knowledge is assumed to be equal to age. 
  \\
  \hspace*{0.2cm} More precisely, our objective is to minimize the total labor cost defined as $C=\int_{z_{\rm min}}^{z_{\rm max}}{\rho^*(z)w(z)dz}$ where $w(z)$ still denotes the cost per employee of age $z$ and $\rho^*(z)$ the concentration of workers of age $z$; under the constraint that the total knowledge $E=\int_{z_{\rm min}}^{z_{\rm max}}{\rho^*(z)zdz}$ is given. This constraint makes sense considering the workers population global knowledge. Knowledge (in other words the experience) rather than hourly workload is a better proxy to describe business needs. 
  \\ \hspace*{0.2cm}Termination is still not allowed. Recalling that $M$ denotes an antiderivative of the attrition rate $\mu$, we show in the appendix \eqref{ap4} that the optimal workforce structure is defined by
$$
  \rho^*(z)=e^{-M(z)}b\mathbf{1}_{z\geq z_0}, \quad \gamma^*(z) =b\delta_{z_0}e^{-M(z)},
$$
   with 
$$
   b=\frac{E}{\int_{z_0}^{z_{\rm max}}{ze^{-M(z)}}dz}, \quad C=Ed(z_0)
$$
and the optimal hiring age $z_0$ is defined by 
$$
    d(z_0)=\min_z\left(d(z)\right), 
$$
where $d(z)=\frac{f(z)}{g(z)}$ can be interpreted as follows: 
\\
$\bullet$  The numerator $f(z)=\int_{z}^{z_{\rm max}}{w(y)e^{-M(y)}dy}$ represents the average  tenure cost of an employee in the firm
\\
$\bullet$ The denominator $g(z)=\int_{z}^{z_{\rm max}}{ye^{-M(y)}dy}$ represents the average knowledge the employee will {have acquired} if hired at age $z$ during its tenure within the firm.
\\

\hspace*{0.2cm} {Minimizing $d$ translates into hiring at an optimal age ($z_0$) and having employees develop their knowledge and expertise within the company. This aligns with standard good management practices \cite{bersin2013predictions}).Three different cases are possible though depending in the cost and the attrition structures:}
\\
$\bullet$ {Case 1: if the minimum is reached in $z_0 \in (z_{\rm min},z_{\rm max})$, then it is optimal to build internally employees' careers starting from the age $z_0$. The firm is here doing long term investments in knowledge workers.}
\\
$\bullet$ {Case 2: if the minimum is reached in $z_{\rm max}$, then it is optimal to hire a pool of experts of maximum experience. However, those experts have to be newly hired each year, and this framework does not take into account the recruiting time and cost as well as the losses induced by this type of disruption. Normally one would expect in such cases to have a minimum reached a little bit before $z_{\rm max}$. This would mean focusing on hiring experts and maximizing their lifecycle within the company. This situation would translate into a problem of succession management.}
\\
$\bullet$ {Case 3: if the minimum is reached in $z_{\rm min}$, then the firm counts on recruiting a high number of young employees, in order to train and keep them until retirement age. }

\paragraph{Examples.}
         
 {In order to provide better illustrations, we choose to apply the above mentioned framework to three different BUs (BU 1, BU 2, BU 3). Those BUs slightly differ from the ones mentioned in the previous sections. Note that the attrition rate is kept at $ \mu=30\%$.Each BU represents a possible real scenario. Note that for confidentiality reasons, no futher specific characteristic will be communicated on the BU set up or composition.}
  \\
  \\
\includegraphics[scale=0.235]{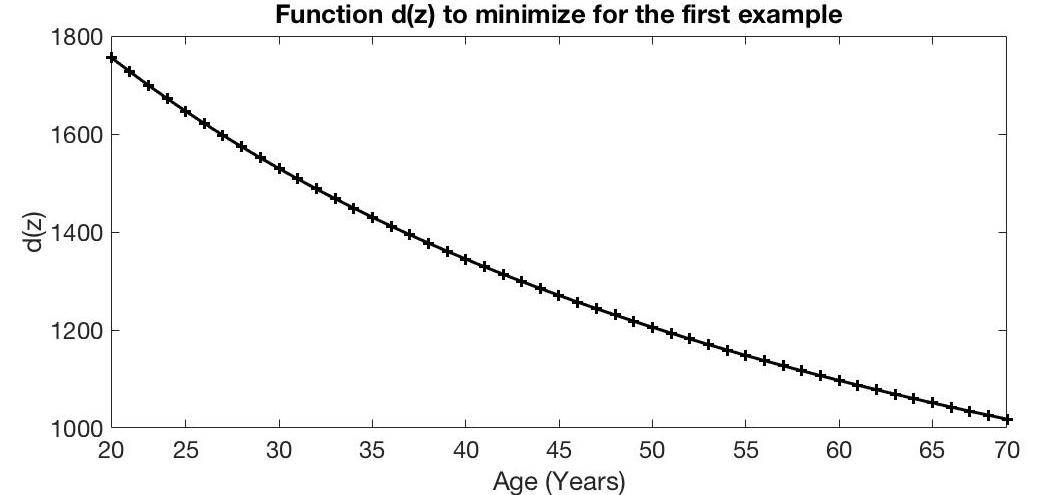}
\includegraphics[scale=0.225]{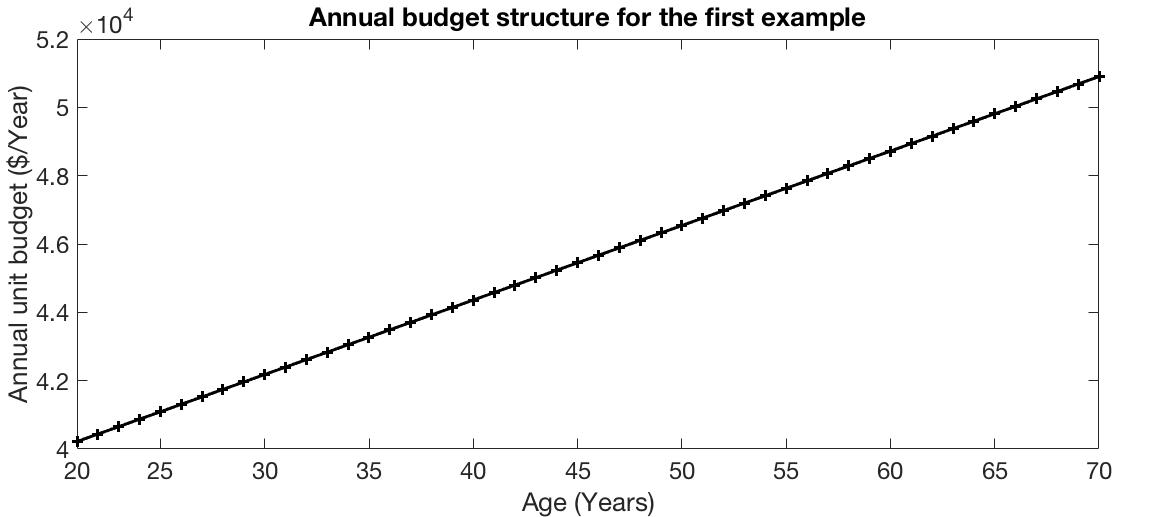}
  \\
  \\[-6mm]
\begin{scriptsize}
\hspace*{0.2cm} { FIG. 9. \textit{Function to minimize $ d(z)$ and budget structure $\omega(z)$ for the BU 1, for which E=3500 years, and, without optimization, average age is 35 years and corresponding labor cost is \$5 million/year (for a total headcount of 100).  The budget is linear, with positive coefficients. }}
\end{scriptsize}
\\
  \\
 \hspace*{3.5cm}\includegraphics[scale=0.25]{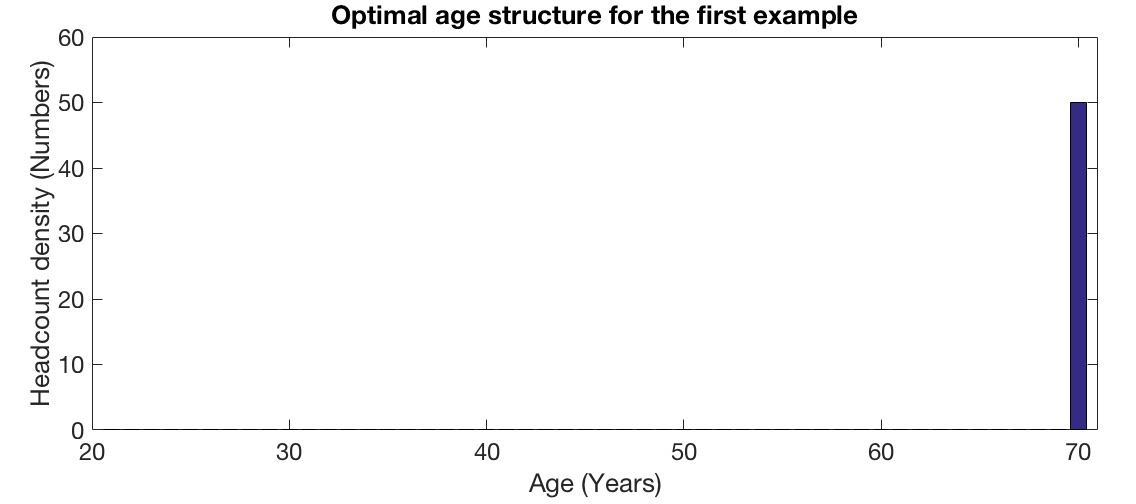}
\\
  \\[-6mm]
\begin{scriptsize}
\hspace*{0.2cm} {FIG. 10. \textit{Optimal age structure for the BU 3. \\ \\}}
\end{scriptsize}
\\
  \\
\includegraphics[scale=0.23]{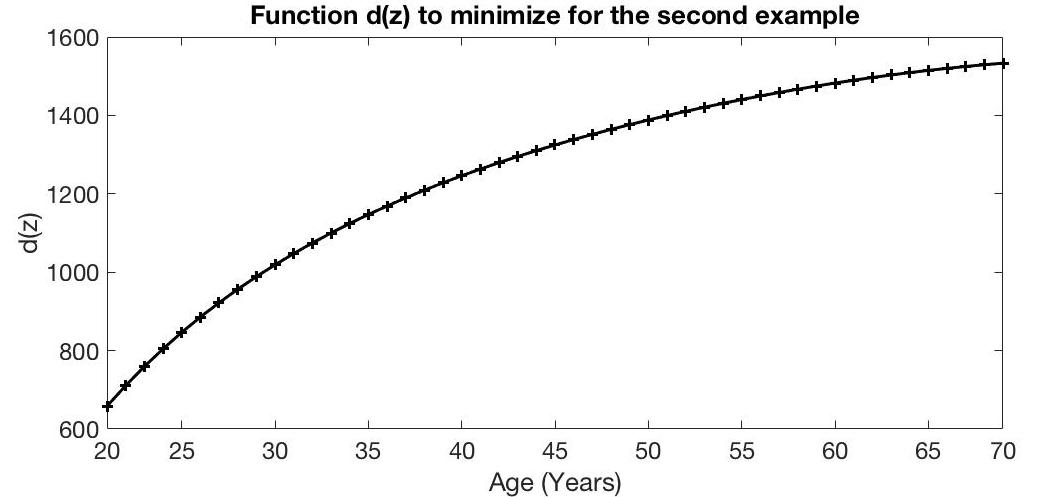}
\includegraphics[scale=0.23]{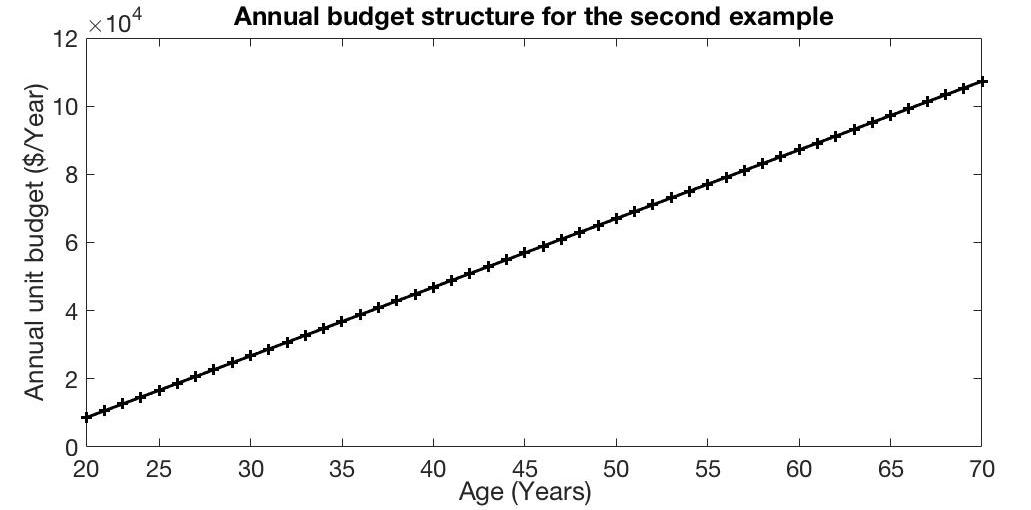}
 \\
   \\[-6mm]
\begin{scriptsize}
\hspace*{0.2cm} { FIG. 11. \textit{Function to minimize $ d(z)$ and budget structure $\omega(z)$ for the BU 2, for which E=3000, and, without optimization, average age is 30 and corresponding labor cost is 2 \$million/year (for a total headcount of 100).  The budget is linear.}}
\end{scriptsize}
 \\
  \\
 \hspace*{2cm}\includegraphics[scale=0.30]{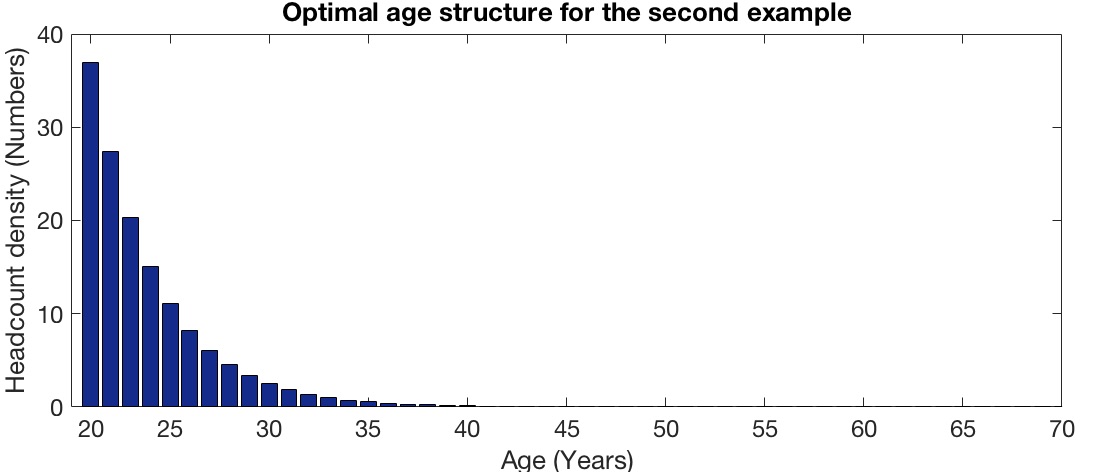}
 \\
  \\[-6mm]
\begin{scriptsize}
\hspace*{0.2cm} {FIG. 12. \textit{Optimal age structure for the BU 2.}}
\end{scriptsize}
 \\
  \\
\includegraphics[scale=0.23]{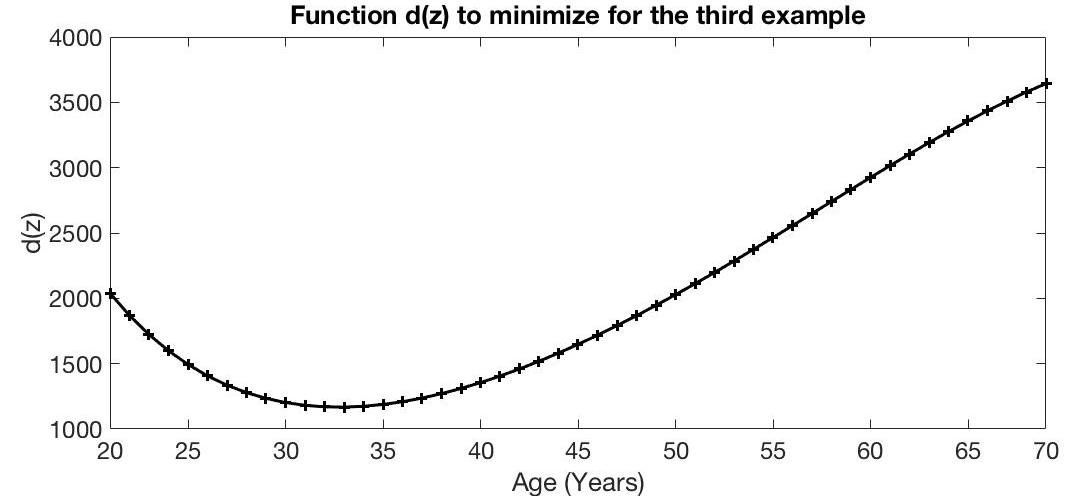}
\includegraphics[scale=0.23]{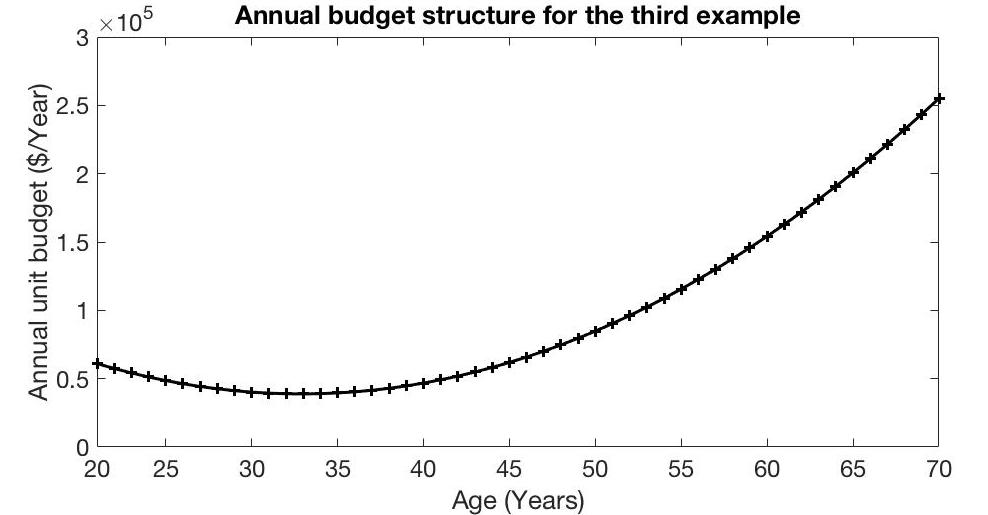}
 \\
   \\[-6mm]
\begin{scriptsize}
\hspace*{0.2cm} { FIG. 13. \textit{Function to minimize $ d(z)$ and budget structure $\omega(z)$ for the BU 3, for which E=3700, and, without optimization, average age is 37 and corresponding labor cost is \$6 million/year (for a total headcount of 100).  The unit budget is a polynomial of degree 2. }}
\end{scriptsize}
\\
  \\
 \hspace*{2cm}\includegraphics[scale=0.30]{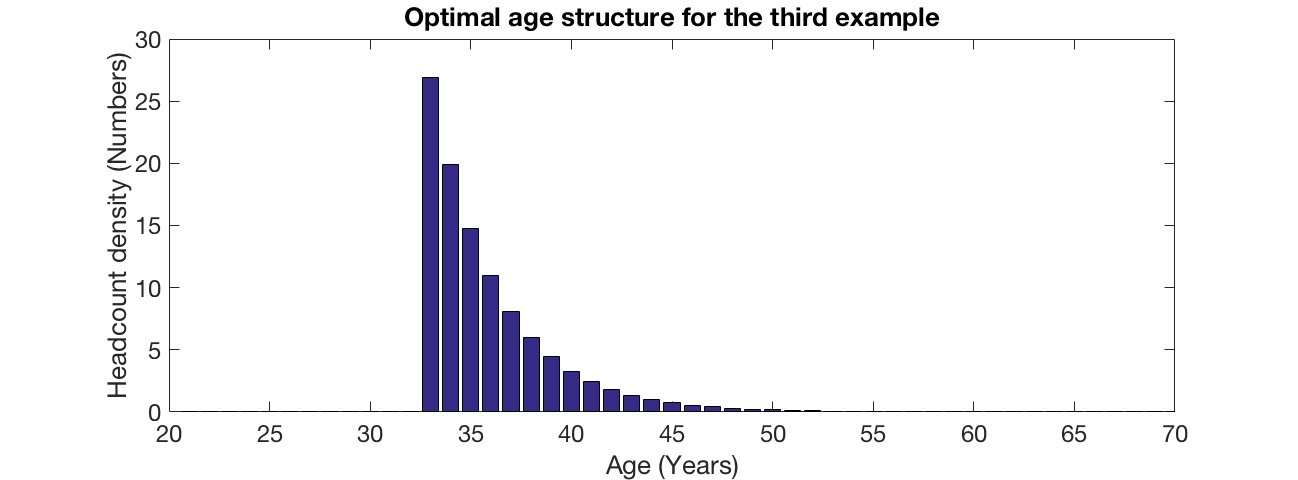}
\\
  \\[-6mm]
\begin{scriptsize}
\hspace*{0.2cm} {FIG. 14. \textit{Optimal age structure for the BU 3. \\ \\}}
\end{scriptsize}
For the BU 1, we can see that the minimum is {at the retirement age} (Figure 9), and we can deduce that the ideal age structure of Figure 10 is 50 people close to retirement (Case 2). This is a typical scenario, where many years of experience are usually required. The optimized labor cost is approximately \$3 million/year, which represents a \$2 million/year saving (approximately 40\% of the total labor cost). However, as we said before, this framework does not take into account the recruiting cost and time to fill, which is not realistic. A suboptimal solution or a framework review should therefore be in order.
\\
\hspace*{0.2cm} For the BU 2, we can see that the minimum is at the minimum age 20 (Figure 11), and we can deduce { the optimal hiring age and the ideal age structure of Figure 12 (Case 1). This happens when the salary gap} between the young and the old employees overtops the associated experience gap. The optimized labor cost is approximately \$1.8 million/year, which represents a \$0.2 million/year saving (approximately 10\% of the total labor cost). We can see that people are hired at 20 years and they progressively leave the company as they age. The average age is 25 (instead of 30 for the non optimized situation), and the total headcount is approximately 120 (instead of 100). 
 \\
\hspace*{0.2cm} For the BU 3, we can see that the minimum is at the age 33 (Figure 13), and we can deduce { the optimal hiring age and the ideal age structure of Figure 14 (Case 1). Here, the most experienced employees are expensive and represent a small proportion of the total workforce, whereas the young ones are less expensive and account for most of the workforce}.  The optimized labor cost is approximately \$5 million/year, which represents a \$1 million/year saving (around 15\% of the total labor cost). We can see that people are hired at 33 years and they progressively leave the company as they age. The average age is 37 (just as in the non optimized situation), and the total headcount is also approximately 100. 
 \\
\hspace*{0.2cm}
{This minimization provides generic solutions to workforce design challenges under experience and cost constraints. The three scenarios that arise from the study described above is in line with the idea of developing one's workforce over the long term,	and should provide ideas to better handle the workforce related business needs assessment of a company.}

\section{Conclusion { and discussion}}

The structured equations framework developed in this paper is a suitable first milestone to get preliminary answers to standard long term workforce concerns such as population stability or the mandatory adaptability of a company hiring policies. This framework can also be leveraged to provide generic solutions to workforce design challenges under experience and cost constraints. So far, we have studied two issues. Firstly, assuming the age profile is known, we have considered hiring strategies able to stabilize the employees population, either based on the total headcount or on a budget constraint. Secondly, we have studied the hiring profile in order to reach an optimal age profile at equilibrium under an experience constraint.
\\ 
\hspace*{0.2cm} { In comparison with discrete stochastic approaches, which represent the main research corpus in SWP, our deterministic continuous approach may appear as a useful complement. We gain here a model which is easy to simulate with fast algorithms, which allows for high flexibility to carry out sensitivity analysis, and which properties are demonstrated thanks to  many theoretical tools  from the field of structured equations. It is particularly well-designed for large populations, in which an averaging effect leads the model to be a good approximation of the equation satisfied in expectancy by a stochastic individual-based approach. However, even for small populations, our equations may be viewed as satisfied \emph{in expectancy} by the agents, thus giving interesting insights on the expected evolution of the population.}
\\
\hspace*{0.2cm} Several limitations to the current paper arise. As mentioned earlier, for small populations, random effects are better represented by stochastic models. For large populations, from a theoretical standpoint, our results and methods do not allow time variations in the attrition nor hired population distribution and the present framework  ignores the workforce transition from one job to another while staying within the same firm. Also, so far the model accounts for only two variables (age and time) and one population class, we will come back on this limitation by also taking into account the tenure in company \cite{SWP2}.
From a practical perspective, the main shortcoming of the study is the lack of productivity function that has been replaced by constraints on experience. Therefore this paper should be considered as a preliminary study case for SWP.
\\ 
 \hspace*{0.2cm} Therefore, a first natural next step would be to optimize the labor costs under population and experience constraints. This type of constraint would be suited to investigate cost-optimal demographic structure for non-knowledge workers. Their overall activity is first determined by workload constraint that is not demographic in nature ({for instance machine workload}) which leads to a population size requirement. Experience would still be important because it represents a knowledge process that cannot be acquired prior to a certain experience threshold. This type of multiple constraints minimization is an extensively researched topic called the linear programming problem. This domain has been pioneered in the 60s (\cite{tyndall1965duality}), and followed by many studies (\cite{pullan1993algorithm,reiland1980optimality}). As another next step, in a continuation of the present analysis, the notion of productivity and a study case on sales representatives could be investigated. Then, the framework could  be expanded to a multi-population framework in order to better represent layers within a company.

\appendix 
  
               
\section{{ Study of} the non-linear equation \eqref{eq1}} 
\label{appen1}
                
{We now perform a mathematical study of the qualitative behavior (stationary states, asymptotic convergence)  of Equation \eqref{eq1},  with respect to the  parameter $\beta$ defined in \eqref{toto} as
$$
  \beta=\int_{z_{\rm min}}^{z_{\rm max}} \left(\int_{z_{\rm min}}^{z}{\gamma(y)e^{-(M(z)-M(y))}dy}\right)dz 
$$
which is  essentially  positively correlated with  the mean coefficient of recruitment.
  We recall the equation~\eqref{eq1} under consideration
\begin{equation*}
\left\{
\begin{array}{l}
 \frac{\partial\rho}{\partial t}(t,z)+\frac{\partial\rho}{\partial z}(t,z)=-\mu (z)\rho(t,z)+\frac{P_t}{1+\alpha P_t^2}\gamma (z), \qquad z_{\rm min}<z<z_{\rm max}, \qquad \alpha >0,
 \\
  \rho(t,z_{\rm min})=0,\\
  \rho(0,z)=\rho^0(z)\geq 0,
\end{array}
\right.
\end{equation*}
 Here, to ensure existence and uniqueness of  a solution  $\rho \in C_b(\mathbb{R}_+,L^1(z_{\rm min},z_{\rm max}))$ of   Equation \eqref{eq1} (\cite{perthame2006transport}), we assume that   $\mu, \gamma$ and $ \rho^0$ are bounded on $[z_{\rm min},z_{\rm max}]$.
\\
We first show that for $\beta \leq 1$, the only  stationary state is zero and when $\beta>1$,  the recruitment is large enough to ensure a unique non vanishing stationary state (see Proposition~\ref{prop1}).  Then, we give a result of asymptotic convergence  (see Propositions~\ref{prop2} and \ref{prop3}). 
 }    
 
\subsection{Existence of steady states}\label{ap}

\begin{proposition}\label{prop1}
         When $\beta\leq1$, the only equilibrium of Equation \eqref{eq1} is zero. When $\beta>1$, the system  \eqref{eq1} admits two equilibrium states: zero and a positive one.
           \end{proposition}
            \begin{proof}
     The stationary states, $\rho_{\rm eq}$, of Equation \eqref{eq1}  are solution of the equation 
   $$
\frac{\mathrm{d}\rho_{\rm eq}}{\mathrm{d}z}(z)=-\mu (z)\rho_{\rm eq}(z)+a^*\gamma (z),
$$
where    $a^*=\frac{P_{\rm eq}}{1+\alpha P_{\rm eq}^{2}}$ { and  $P_{eq}$ the total headcount given by} $P_{\rm eq}=\int \limits_{z_{\rm min}}^{z_{\rm max}}{\rho_{\rm eq}(z)dz}$. { As 
 $$\displaystyle \frac{\mathrm{d}}{\mathrm{d}z}\left(\rho_{\rm eq}(z)e^{M(z)}\right)=a^*\gamma (z)e^{M(z)},$$ we obtain 
 $$
 \rho_{\rm eq}(z)=\int_{z_{\rm min}}^{z}{a^*\gamma(y)e^{-(M(z)-M(y))}dy}=\frac{P_{\rm eq}}{1+\alpha P_{\rm eq}^{2}}\int_{z_{\rm min}}^{z}{\gamma(y)e^{-(M(z)-M(y))}dy}.
 $$
 Integrating the above equation, we find that $P_{\rm eq}$ must  satisfy the equation 
 \begin{equation}\label{eqPstatio} 
 P_{\rm eq}= \beta \frac{P_{\rm eq}}{1+\alpha P_{\rm eq}^{2}} \cdotp
 \end{equation}
 Now, either   $\beta \leq 1$, and the only possible solution of \eqref{eqPstatio} is $P_{\rm eq}=0$ either, $\beta>1$ and there are two solutions of the Equation \eqref{eqPstatio} given by 
 $$
 P_{\rm eq}=0 \quad  \hbox{ and } \quad P_{\rm eq}=  \sqrt{  \frac{ \beta -1}{ \alpha}},
 $$
 which ends the proof of Proposition~\ref{prop1}.} \hfill \qed 
    \end{proof}
    
\subsection{Asymptotic behavior}
   
\begin{proposition}[case $\beta < 1$] \label{prop2}
If $\beta < 1$, the total population $P(t)$ goes exponentially fast to 0 and 
\begin{equation}\label{lrho}
\lim\limits_{t\to\infty} \Vert \rho(t) \Vert_{L^\infty \big(z_{\rm min},\, z_{\rm max}\big)}= 0.
\end{equation}
 \end{proposition}
  \begin{proof}
 {Using the characteristics, we  find that for  $s \in [t-(z_{\rm max}-z_{\rm min}),t]$ and  $t \geq (z_{\rm max}-z_{\rm min})$,
 \begin{equation}\label{lPrho}
  \rho(t,z)= e^{-M(z)} \int_{z_{\rm min}}^{z} a(t-z+\tau) \gamma(\tau)  e^{M(\tau) } d\tau, \quad z \in [z_{\rm min},z_{\rm max}], \quad  a(s)=  \frac{P(s)}{1+ \alpha P^2(s)} \cdotp
  \end{equation}
 Integrating from $z_{\rm min}$ to $z_{\rm max}$, and using that $a$ is a  bounded function, we obtain that there exists $C>0$  such that 
 \begin{equation}\label{idP}
  P(t) =  \int_{z_{\rm min}} ^{z_{\rm max}}  e^{-M(z) }\int_{z_{\rm min}}^{z} a(t-z+\tau) \gamma(\tau)  e^{M(\tau) } d\tau dz \leq C.
  \end{equation}
 Moreover, as $\alpha > 0$,  we have $a(t)=  \frac{P(t)}{1+ \alpha P^2(t)}  \leq P(t)$, and so 
$$
 P(t) \leq \beta \sup_{s \in  [t-(z_{\rm max}-z_{\rm min}),t]} P(s).
 $$
 Now, we observe that, by iteration, for all integer $n \geq1$, for  all $t \geq (n+1)(z_{\rm max}-z_{\rm min})$,
$$
P(t) \leq \beta^n   \sup_{s \in  [t-n(z_{\rm max}-z_{\rm min}),t]} P(s) \leq \beta^n \sup_{s \geq (z_{\rm max}-z_{\rm min})} P(s) \leq C \beta^n
$$  
which leads to exponential convergence of $P$ to $0$ because $ \beta<1$. The proof of estimate \eqref{lrho} is then a direct consequence of identity \eqref{lPrho}, which ends the proof of Proposition \ref{prop2}.}  \hfill \qed 
 \end{proof}         
\begin{proposition}[case $\beta > 1$]\label{prop3}
Assume that $\rho^0(z) \neq 0$ and $1<\beta<9$. Then,  $P(t)$ goes exponentially fast to $P_{\rm eq}>0$ and, denoting $\rho_{\rm eq}$ the positive equilibrium state, we have:
$$
\lim\limits_{t\to+\infty} \Vert \rho(t) - \rho_{\rm eq}\Vert_{L^\infty (z_{\rm min},\, z_{\rm max})}= 0.
$$
\end{proposition}
{
\begin{remark}
The condition $\beta <9$ seems to be a technical condition which ensures that the nonlinearity of the equation is not too strong, allowing us in the proof to use a contraction type  argument.  However, numerically, we observe that the solution converges to a stationary state even if $\beta \geq 9$. Hence, we expect that Proposition~\ref{prop3} holds also for $\beta \geq 9$.
\end{remark}}
 \begin{proof}
 {
  As $\beta=1+ \alpha P_{\rm eq}^2$, with identity \eqref{idP}, we obtain that for  $t \geq z_{\rm max}-z_{\rm min}$ }
  $$ 
  |P(t)-P_{\rm eq}|  \leq  \sup_{s \in   [t-(z_{\rm max}-z_{\rm min}),t]} |a(s)-a^*| (1+ \alpha P_{\rm eq}^2),
  $$
  which leads us to 
  $$ 
   |P(t)-P_{\rm eq}|  \leq   \sup_{s \in   [t-(z_{\rm max}-z_{\rm min}),t]} \left|  \frac{P(s)  (1+ \alpha P_{\rm eq}^2)- P_{\rm eq}(1+ \alpha P^2(s))}{1+ \alpha P^2(s)} \right |,
   $$
  and
  $$
  |P(t)-P_{\rm eq}|  \leq   \sup_{s \in   [t-(z_{\rm max}-z_{\rm min}),t]} |  P(s)-P_{\rm eq}|  \left|  \frac{(1- \alpha PP_{\rm eq})}{1+ \alpha P^2(s)} \right|.
  $$
Denoting $c= \sqrt{ \beta -1}$, we have $ \displaystyle P_{\rm eq}= \frac{c}{ \sqrt{ \alpha}}, $
 which, with $f_c(x)= \frac{ 1-cx}{1+x^2}$, we obtain
 $$
 |P(t)-P_{\rm eq}|  \leq   \sup_{s \in   [t-(z_{\rm max}-z_{\rm min}),t]} |  P(s)-P_{\rm eq}|  \sup_{s \in   [t-(z_{\rm max}-z_{\rm min}),t]}  |f_c\left(\sqrt{\alpha} P(s)\right) |.
 $$
 { This ends the proof of Proposition \ref{prop3},   using the same arguments as in Proposition \ref{prop2}, assuming that for $1<\beta<9$,  
  $$
   \sup_{s \in   [t-(z_{\rm max}-z_{\rm min}),t]}  |f_c(\sqrt{\alpha} P(s)) | \leq C_1 <1,
   $$
   which is the purpose of the next lemma.}
\hfill \qed 
 
  \end{proof}
  \begin{lemma}\label{lem1}
 Under assumptions of Proposition~\ref{prop3}, the following estimate holds
  $$
  \sup\limits_{s\geq z_{\rm max}-z_{\rm min}} \vert f_c\big(\sqrt{\alpha} P(s)\big)\vert <1.
  $$ 
  \end{lemma}
  \begin{proof}
{ With  the assumptions of Proposition~\ref{prop3}, we have  $c<2\sqrt{2},$ and so $f_c(x) <1$ for all $x>0$.  Hence the proof of Lemma~\ref{lem1} can be reduced to show that there exists  $\bar{P}>0$ such that 
\begin{equation}\label{propP}
P(s)\geq \bar{P}  \hbox{ for all } s\geq z_{\rm max} -z_{\rm min}.
\end{equation}
 To do this, let us first prove that  for any $t\geq 0,$ $P(t)>0$.  Let us observe that, If  $P(t)>0$ for $t\leq z_{\rm max}-z_{\rm min}$, then, using the formula \eqref{eqPstatio} on $P$ we can deduce that $P(t)>0$ for all $t$. Indeed,  if  there exists $t_0 >z_{\rm max}-z_{\rm min}$  for which $P(t_0)=0$ and $P(t)>0$ for all $ t <t_0$, then, we have 
   $$
    P(t_0) =  \int_{z_{\rm min}} ^{z_{\rm max}}  e^{-M(z) }\int_{z_{\rm min}}^{z} a(t_0-z+\tau) \gamma(\tau)  e^{M(\tau) } d\tau dz
    $$
     with $a(t_0-z+\tau)>0$ everywhere except in $\tau=z$, which would bring us to $P(t_0)>0$, in direct contradiction with $P(t_0)=0$.} So it remains to show that $P(t)>0$ for $t\leq z_{\rm max}-z_{\rm min}$. For $t\leq z_{\rm max}-z_{\rm min}$, we have 
$$
     P(t)=  \int_{z_{\rm min}}^{z_{\rm min}+t} \rho(t,z)dz+   \int_{z_{\rm min}+t}^{z_{\rm max}} \rho(t,z)dz.
$$
   Yet, the method characteristics gives the representation formulas
$$
\rho(t,z)= e^{-M(t)} \rho(0,z-t) + e^{-M(t)} \int_{0}^{t} e^{M(s)}a(s)  \gamma(z-t+s)ds \;  \hbox{ for  }\;  t \leq z-z_{\rm min},
$$
 $$
 \rho(t,z)=  \int_{z_{\rm min}}^{z} e^{-M(z)+M(\tau)} \gamma(\tau) a(t-z+\tau)d\tau \quad \hbox{ for  }\;  t \geq z-z_{\rm min}.
 $$
 {By the same reasoning we obtain that $P(t)>0$  for  $t\leq z_{\rm max}-z_{\rm min}$, because $P(0)>0$, and hence $P(t)>0$  for all $t \geq 0$.  Let us now conclude estimate \eqref{propP}. We know that for $t \geq z_{\rm max}-z_{\rm min}$, }
 \begin{equation}\label{eqP}
P(t) \geq \beta \min_{s \in [t-(z_{\rm max}-z_{\rm min}),t]} a(P(s)).
\end{equation}
Since $\beta >1$, we can find $\bar{P}>0$  and $\varepsilon >0$ with the following properties:
\begin{equation}\left\{\begin{array}{ll}
&  \ \beta a(s)  \geq s(1+ \varepsilon), \quad \forall   s \leq \bar{P}, \\ 
&  \ a \text{ is strictly increasing on }(0, \bar{P}), \\ 
&  \  \beta  \inf_{s \in (\bar{P},\frac{1}{2 \sqrt{\alpha} }\beta)}a(s)  \geq \bar{P} (1+ \varepsilon).
\end{array}\right. \label{P3}
\end{equation}
Now, if $P(t) \geq \bar{P}$ for  $t \geq z_{\rm max}-z_{\rm min}$, our lower bound is proved. Else, there exists $t_0\geq z_{\rm max}-z_{\rm min}$ a time for which $P(t_0) \leq \bar{P}.$
In this case,  we have  
$$
\bar{P} \geq P_{\rm inf}(t_0)  \hbox{ where }  P_{\rm inf}(t_0)= \min_{s \in [t_0-(z_{\rm max}-z_{\rm min}),t_0]} P(s)>0.
$$
Let us first prove that 
\begin{equation}\label{PP2}
P(t_0) \geq P_{\rm inf}(t_0) (1+\varepsilon).
\end{equation}
{
Denoting  ${}^c A$ as the complement of $A$, we write 
$$ \{s \in [t_0-(z_{\rm max}-z_{\rm min}),t_0] \} = A \cup  A{}^c  \hbox{  with } A= \{  s \in [t_0-(z_{\rm max}-z_{\rm min}),t_0] \hbox{ such that }  P(s) \geq \bar{P} \}.$$}
Using~\eqref{eqP}, we obtain
$$ P(t_{0}) \geq   \min( \min_{s \in A} \beta a(P(s)) , \min_{s \in {}^cA}  \beta a(P(s)).$$
 Using the  first part of ~\eqref{P3}, we deduce that 
  $$\beta \min_{s \in {}^cA} a(P(s)) \geq   P_{\rm inf}(t_0) (1+\varepsilon),$$
  and using the second part of ~\eqref{P3}, we obtain that 
  $$\beta \min_{s \in A} a(P(s)) \geq  \bar{P}(1+\varepsilon) \geq  P_{\rm inf}(t_0) (1+\varepsilon)$$
  and so  estimate \eqref{PP2} holds. { Let us now prove that, for all  $t \geq t_0$, 
$ \displaystyle P(t) >  P_{\rm inf}(t_0).$} If it is not the case, since  
  $ \displaystyle P(t_{0}) >P_{\rm inf}(t_0),$
 we would find   $t_1>t_0$  for which
  $$
  P(t_1)= P_{\rm inf}(t_0)\hbox{ and  }   \min_{s\in [t_1-(z_{\rm max}-z_{\rm min}),t_1]}P(s) \geq P_{\rm inf}(t_0) .
  $$
  Combining again~\eqref{eqP} and~\eqref{P3} , we would obtain
  $$
  P(t_{1}) \geq  P_{\rm inf}(t_0)(1+ \varepsilon),
  $$
which is in contradiction with $P(t_1)=P_{\rm inf}(t_0)$.This completes the proof of  of Lemma \ref{lem1}.
  
\hfill \qed 
\end{proof}
 
\subsection{Numerical method}
\label{appen}

{The examples of subsection~\ref{numex1} have been obtained using a numerical method that we explain now.
 We discretize the time  interval $[0,T]$ using uniform subintervals of size $\delta_t$,  and  the age interval $[z_{\rm min},z_{\rm max}]$ is discretized using uniform subintervals of size $\delta_z$. The, we set $z_j=z_{\rm min}+ j\delta z$, $t_k=k\delta t$ and $\rho (t_k,z_j)=\rho_j^k$, $\mu (z_j)=\mu_j$, $\gamma (z_j)=\gamma_j$ (with $\rho_0^k=0$ and $\rho_j^0$ given).
  We choose the following semi-implicit scheme, for $j\ge 1$ and $k\ge 0$,   
 $$
    \overbrace{\frac{\rho_j^{k+1}-\rho_j^k}{\delta t}+\frac{\rho_j^k-\rho_{j-1}^k}{\delta z}}^{Workforce \, aging}=-\overbrace{\mu_j\rho_j^{k+1}}^{Attrition}+\overbrace{\frac{P_k}{1+\alpha P_k^2}\gamma_j}^{Hiring},
   \hbox{ where } P_k=\sum_{j}^{}\rho_j^k.
$$
   Let us mention that, according to the Courant-Friedrichs-Levi condition (\cite{bouchut2004nonlinear,leveque1992numerical}), we have to impose the condition $\frac{\delta t}{\delta z}\leq1$ for stability. }  
  
\section{{ Study of}  the linear model ~\eqref{eq3}} 
\label{ap2}

{  We determine the set of possible stationary states of Equation~\eqref{eq3}  and study the asymptotic behavior of the solution. We  prove that the steady states are all proportional to  a particular  positive state 
and  show that, under certain assumptions,   the solution converges to a non zero stationary state. The tools used involve, in particular, entropy methods (\cite{michel2005general,perthame2006transport}).} We recall the equation~\eqref{eq3} under study 
     \begin{equation*}
        \left\{
\begin{array}{l} 
  \frac{\partial\rho}{\partial t}(t,z)+\frac{\partial\rho}{\partial z}(t,z)=-\mu (z)\rho(t,z)+h([\rho])\gamma (z), \qquad z_{\rm min}<z<z_{\rm max}, \\
  \rho(t,z_{\rm min})=0,\\
  \rho(0,z)=\rho^0(z)\geq 0,
\end{array}
\right.
     \end{equation*}
 with
 $$
 h([\rho])=\frac{\int_{z_{\rm min}}^{z_{\rm max}}{\omega(z)\mu (z)\rho(t,z)dz}+\omega (z_{\rm max})\rho (t,z_{\rm max})-\int_{z_{\rm min}}^{z_{\rm max}}{\rho(t,z)\frac{\partial\omega}{\partial z}(z)dz}}{\int_{z_{\rm min}}^{z_{\rm max}}{\omega(z)\gamma (z)dz}}.  
 $$
 {
To ensures that $\rho \in C^1_b(\mathbb{R}_+ \times [z_{\rm min},z_{\rm max}])$ (see \cite{perthame2006transport}),  we assume that  
\begin{equation}\label{hypdo}
 \gamma, \ \mu  \in C([z_{\rm min}, _{max}]) ,\quad  \rho^0 \in C^1([z_{\rm min},z_{\rm max}]), \quad \omega \in C^1([z_{\rm min},z_{\rm max}]).
\end{equation} }
  
\subsection{{ Study of} stationary states}

\begin{proposition}\label{propB1}
 {The set $S$ of  steady states  of Equation \eqref{eq3} is given by
  $$ S= \{ h: [z_{\rm min}, z_{\rm max}] \to \R^+ \hbox{ such that }  h(z)=C \int_{z_{\rm min}}^{z}{ \gamma(y)e^{-(M(z)-M(y))}dy}, \hbox{ where } C \in \R^+\}. $$}
\end{proposition}
\begin{proof}
  The equilibrium equation is     
  $$
  \frac{\mathrm{d}\rho_{\rm eq}}{\mathrm{d}z}(z)=-\mu (z)\rho_{\rm eq}(z)+h^*\gamma (z), $$
  with 
  \begin{equation}\label{h*}
  h^*=\frac{\int_{z_{\rm min}}^{z_{\rm max}}{\omega(z)\mu (z)\rho_{\rm eq}(z)dz}+\omega (z_{\rm max})\rho_{\rm eq} (z_{\rm max})-\int_{z_{\rm min}}^{z_{\rm max}}{\rho_{\rm eq}(z)\frac{\partial\omega}{\partial z}(z)dz}}{\int_{z_{\rm min}}^{z_{\rm max}}{\omega(z)\gamma (z)dz}} \cdotp
  \end{equation}
   So, { with the boundary condition   $\rho(z_{\rm min})=0$}, we obtain {that $\rho_{eq}$ should satisfies the implicit formula}
   \begin{equation}\label{imp}
   \rho_{\rm eq}(z)=\int_{z_{\rm min}}^{z}{h^*\gamma(y)e^{-(M(z)-M(y))}dy}.
   \end{equation}
   { Observing that 
   $$\omega(z)\gamma (z)= \omega(z) e^{-M(z)} \left( \int_{z_{\rm min}}^{z}\gamma(y)e^{M(y)}dy \right)',$$ 
   an integration by part,  shows that 
   $$\bea
 & \int_{z_{\rm min}}^{z_{\rm max}}{\omega(z)\gamma (z)dz}   
  \\[5pt]
   &=\int_{z_{\rm min}}^{z_{\rm max}}{   \left((\omega(z)\mu(z)-\omega^{'}(z))\left(\int_{z_{\rm min}}^{z}{\gamma(y)e^{-(M(z)-M(y))}dy}\right)+\omega(z_{\rm max})\gamma(z)e^{-(M(z_{\rm max})-M(z))}   \right)dz},
   \eea
   $$  
hence the constraint \eqref{h*} always holds which ends the proof of Proposition \ref{propB1}. \hfill \qed }
   \end{proof}

\subsection{Asymptotic behavior}

   { In this section, we assume that the cost  $w$ satisfies the following assumptions
\begin{equation}\label{asumpw}
   w \geq 0, \qquad w' \leq \mu w,  \qquad w(z_{\rm max})>0.
\end{equation}
   In  order to prove that the hiring strategy under consideration converges, we introduce some notations. We rewrite the equation as:
$$
\frac{\partial\rho}{\partial t}(t,z)+\frac{\partial\rho}{\partial z}(t,z)=-\mu (z)\rho(t,z)+A\gamma (z)\rho(t,z_{\rm max})+\gamma (z)\int_{z_{\rm min}}^{z_{\rm max}}{B(y)\rho(t,y)dy},  
$$ 
with, using assumption \eqref{asumpw},  }
$$
  A=\frac{\omega (z_{\rm max})}{\int_{z_{\rm min}}^{z_{\rm max}}{\omega(z)\gamma (z)dz}} > 0, \qquad B(y)= \frac{\mu (y)\omega(y)-\frac{\partial\omega}{\partial y}(y)}{\int_{z_{\rm min}}^{z_{\rm max}}{\omega(z)\gamma (z)dz}}\geq 0.
$$
{In the following, $\rho_{\rm eq}$ denotes the particular stationary solution given by 
\begin{equation}\label{choicep}
\rho_{\rm eq}= \int_{z_{\rm min}}^{z}\gamma(y)e^{-(M(z)-M(y))}dy.
\end{equation}
The following proposition holds.}
 \begin{proposition}\label{proplemme}
 {
     Assume  \eqref{asumpw}, \eqref{hypdo}, that $w \gamma >0$ on $[z_{\rm min},z_{\rm max}]$, let  $\rho^0$ a positive initial data and  $\rho_{eq}$ given by \eqref{choicep}.  Then, 
$$
\lim_{t \to +\infty}\int_{z_{\rm min}}^{z_{\rm max}} \left( \frac{\rho(t,z)-m\rho_{\rm eq}(z)}{\rho_{eq}(z)} \right)^2\omega(z)  \gamma(z)dz =0, \hbox{ where } m=\frac{\int_{z_{\rm min}}^{z_{\rm max}}\rho^0(z) \omega(z)dz}{\int_{z_{\rm min}}^{z_{\rm max}}\rho_{\rm eq}(z) \omega(z)dz} \cdotp
$$}
\end{proposition}
\begin{proof}
{The proof of  Proposition \ref{proplemme} is based  on a  general relative entropy inequality stated in the following lemma  (see  \cite{perthame2006transport} for general entropy methods).} 

\begin{lemma}\label{propentro}
{
Assume that \eqref{asumpw} holds and let $\rho$ be a solution of \eqref{eq3} and let $\rho_{eq}$ given by \eqref{choicep}. Then, the  following estimate holds
$$
\frac{\mathrm{d}}{\mathrm{d}t}\int_{z_{\rm min}}^{z_{\rm max}}\omega(z)\rho_{\rm eq}(z)\left(\frac{\rho(t,z)}{\rho_{\rm eq}(z)}\right)^2 dz \leq - A \rho_{\rm eq}(z_{\rm max})  \int_{z_{\rm min}}^{z_{\rm max}} \omega(z)\gamma(z)  \left(\frac{\rho(t,z_{\rm max})}{\rho_{\rm eq}(z_{\rm max})}- \frac{\rho(t,z)}{\rho_{\rm eq}(z)}\right)^2dz.$$}
   \end{lemma}

\begin{remark}
{
Let us mention that even if $\rho_{eq}(z_{\rm min})=0$, $\displaystyle \frac{\rho(t,z_{\rm min})}{\rho_{eq}(z_{\rm min})}$ is well defined. Indeed, 
$$\frac{\rho(t,z_{\rm min})}{\rho_{eq}(z_{\rm min})} = \lim_{z \to z_{\rm min}} \frac{\rho(t,z)}{\rho_{eq}(z)} = \frac{\partial_z\rho(t,z_{\rm min})}{\partial_z \rho_{eq}(z_{\rm min})}=
A \rho(t,z_{\rm max})+ \int_{z_{\rm min}}^{z_{\rm max}}B(y) \rho(t,y)dy.$$}
\end{remark}
{
Let us for instance assume that Lemma \ref{propentro} holds and let us finish the proof of Proposition \ref{proplemme}.  Let $\rho$ be a solution of \eqref{eq3} with a nonnegative initial data.  We set 
$$n(t,z)=  \rho(t,z)-m\rho_{\rm eq}(z), \hbox{ where } m=\frac{\int_{z_{\rm min}}^{z_{\rm max}}\rho^0(z) \omega(z)dz}{\int_{z_{\rm min}}^{z_{\rm max}}\rho_{\rm eq}(z) \omega(z)dz} \cdotp $$
Then, $n$ is solution of Equation \eqref{eq3}, and so by Lemma \ref{propentro},
$$\frac{\mathrm{d}}{\mathrm{d}t}\int_{z_{\rm min}}^{z_{\rm max}}\omega(z)\rho_{\rm eq}(z)\left(\frac{n(t,z)}{\rho_{\rm eq}(z)}\right)^2 dz \leq - A \rho_{\rm eq}(z_{\rm max})  \int_{z_{\rm min}}^{z_{\rm max}} \omega(z)\gamma(z)  \left(\frac{n(t,z_{\rm max})}{\rho_{\rm eq}(z_{\rm max})}- \frac{n(t,z)}{\rho_{\rm eq}(z)}\right)^2dz.$$
 Because, $\displaystyle A\rho_{eq}(z_{\rm max})>0$, we deduce, integrating in time the above equation, that
$$\int_{t=0}^{+\infty} \int_{z_{\rm min}}^{z_{\rm max}} \omega(z)\gamma(z)  \left(\frac{n(t,z_{\rm max})}{\rho_{\rm eq}(z_{\rm max})}- \frac{n(t,z)}{\rho_{\rm eq}(z)}\right)^2dz dt<+\infty.$$
As  we have assume \eqref{hypdo}, we have
  $$  \int_{z_{\rm min}}^{z_{\rm max}} \omega(z)\gamma(z)  \left(\frac{n(t,z_{\rm max})}{\rho_{\rm eq}(z_{\rm max})}- \frac{n(t,z)}{\rho_{\rm eq}(z)}\right)^2dz \in \mathcal{C}^{1}_b(\R^+), $$
we deduce that 
$$\lim_{t \to +\infty} \int_{z_{\rm min}}^{z_{\rm max}} \omega(z)\gamma(z)  \left(\frac{n(t,z_{\rm max})}{\rho_{\rm eq}(z_{\rm max})}- \frac{n(t,z)}{\rho_{\rm eq}(z)}\right)^2dz=0.$$
As   $\displaystyle    \left(\frac{n(t,z_{\rm max})}{\rho_{\rm eq}(z_{\rm max})}- \frac{n(t,z)}{\rho_{\rm eq}(z)}\right)^2 \in \mathcal{C}^{1}(\R^+ \times [z_{\rm min},z_{\rm max}]),$
 we deduce that for all $z \in [z_{\rm min}, z_{\rm max}]$ such that $w\gamma(z) >0$    
$$\frac{n(t,z_{\rm max})}{\rho_{\rm eq}(z_{\rm max})}= \lim_{t \to +\infty} \frac{n(t,z)}{\rho_{eq}(z)} \cdotp$$
If $w\gamma >0$ on $[z_{\rm min},z_{\rm max}],$ as $\displaystyle  \int_{z_{\rm min}} ^{z_{\rm max}}w(z) n(t,z)dz=0$, we obtain that  $\lim_{t \to +\infty} n(t,z)=0$ for all  $z \in [z_{\rm min}, z_{\rm max}]$,  
which gives Proposition \ref{proplemme}.}
\hfill \qed
\end{proof}

\vspace{0,5cm}

\noindent{\bf Proof of  Lemma \ref{propentro}.}
{ Let $H$ be a convex function. One can computate
 $$
 \frac{\partial}{\partial t}\left(H\left( \frac{ \rho(t,z)}{ \rho_{\rm eq}(z)}\right)\right)+\frac{\partial}{\partial z}\left( H\left(\frac{ \rho(t,z)}{ \rho_{\rm eq}(z)}\right)\right)=
 $$
 $$
 H^{'}\left(\frac{ \rho(t,z)}{ \rho_{\rm eq}(z)}\right)\left(\int_{z_{\rm min}}^{z_{\rm max}}{ \gamma(z)B(y) \frac{ \rho_{\rm eq}(y)}{ \rho_{\rm eq}(z)}\left( \frac{ \rho(t,y)}{ \rho_{\rm eq}(y)}- \frac{ \rho(t,z)}{ \rho_{\rm eq}(z)}\right)dy}+A \gamma(z) \frac{ \rho_{\rm eq}(z_{\rm max})}{ \rho_{\rm eq}(z)}\left(\frac{ \rho(t,z_{\rm max})}{ \rho_{\rm eq}(z_{\rm max})}- \frac{ \rho(t,z)}{ \rho_{\rm eq}(z)}\right)\right). 
 $$ 
We also have
$$\frac{\partial}{\partial t}\left(\omega \rho_{\rm eq}(z)\right)+\frac{\partial}{\partial z}\left(\omega \rho_{\rm eq}(z)\right)  =A\omega(z)\rho_{\rm eq}(z_{\rm max})+\int_{z_{\rm min}}^{z_{\rm max}}{(B(y)\rho_{\rm eq}(y)\gamma(z)\omega(z)-B(z)\rho_{\rm eq}(z)\gamma(y)\omega(y))dy},
   $$
 then 
  $$
  \bea 
  \frac{\partial}{\partial t}&\left( \omega(z)\rho_{\rm eq}(z)H\left(\frac{ \rho(t,z)}{ \rho_{\rm eq}(z)}\right)\right)+\frac{\partial}{\partial z}\left( \omega(z)\rho_{\rm eq}(z)H\left(\frac{ \rho(t,z)}{ \rho_{\rm eq}(z)}\right)\right)
   \\[5pt]
   &=H\left(\frac{\rho(t,z)}{\rho_{\rm eq}(z)}\right)\left(\int_{z_{\rm min}}^{z_{\rm max}}{(B(y)\rho_{\rm eq}(y)\gamma(z)\omega(z)-B(z)\rho_{\rm eq}(z)\gamma(y)\omega(y))dy} \right)
   \\[5pt]
     &+\omega(z) \gamma(z)H^{'}\left(\frac{\rho(t,z)}{\rho_{\rm eq}(z)}\right)\left(\int_{z_{\rm min}}^{z_{\rm max}}{B(y)\frac{ \rho_{\rm eq}(y)}{ \rho_{\rm eq}(z)}\left( \frac{ \rho(t,y)}{ \rho_{\rm eq}(y)}- \frac{ \rho(t,z)}{ \rho_{\rm eq}(z)}\right)dy}\right)
      \\[5pt]
      &+H\left(\frac{\rho(t,z)}{\rho_{\rm eq}(z)}\right)A\gamma(z)\omega(z)\rho_{\rm eq}(z_{\rm max})+H^{'}\left(\frac{\rho(t,z)}{\rho_{\rm eq}(z)}\right)\gamma(z)\omega(z)A\rho_{\rm eq}(z_{\rm max})\left(\frac{ \rho(t,z_{\rm max})}{ \rho_{\rm eq}(z_{\rm max})}- \frac{ \rho(t,z)}{ \rho_{\rm eq}(z)}\right). 
         \eea
$$  
We now integrate in $z$ and  obtain  that
$$
\frac{\mathrm{d}}{\mathrm{d}t}\left(\int_{z_{\rm min}}^{z_{\rm max}}{\omega(z)\rho_{\rm eq}(z)H\left(\frac{\rho(t,z)}{\rho_{\rm eq}(z)}\right)dz}\right)=-D_1^H(t)-D_2^H(t),$$
 where 
 $$D_1^H(t)=\int_{z_{\rm min}}^{z_{\rm max}}{\int_{z_{\rm min}}^{z_{\rm max}}{\omega\gamma(z)B\rho_{\rm eq}(y)\left (H\left(\frac{\rho(t,y)}{\rho_{\rm eq}(y)}\right)-H\left(\frac{\rho(t,z)}{\rho_{\rm eq}(z)}\right)-H^{'}\left(\frac{\rho(t,z)}{\rho_{\rm eq}(z)}\right)\left(\frac{\rho(t,y)}{\rho_{\rm eq}(y)}-\frac{\rho(t,z)}{\rho_{\rm eq}(z)}\right)\right)dydz}} $$ and 
 $$D_2^H(t)=A\rho_{\rm eq}(z_{\rm max}){\int_{z_{\rm min}}^{z_{\rm max}}{\omega \gamma(z)\left(H\left(\frac{\rho(t,z_{\rm max})}{\rho_{\rm eq}(z_{\rm max})}\right)-H\left(\frac{\rho(t,z)}{\rho_{\rm eq}(z)}\right)-H^{'}\left(\frac{\rho(t,z)}{\rho_{\rm eq}(z)}\right)\left(\frac{\rho(t,z_{\rm max})}{\rho_{\rm eq}(z_{\rm max})}-\frac{\rho(t,z)}{\rho_{\rm eq}(z)}\right) \right)}dz}. $$
 Using that  for a convex function $H(x)-H(y) -H'(y)(x-y) \geq 0$, we deduce that $D_1^H \leq 0$ and $D_2^H \leq 0$.  Taking $H(x)=x^2$ we obtain Lemma \ref{propentro}. 
 }
      \hfill \qed 

\subsection{Numerical method} 
\label{ap3}
    
      We take the same notations as in \ref{appen}. This numerical method is applied to the examples of subsection~\ref{numex2}. In order to simplify some calculations, we choose to determine the concentration of workers of age $z$ at time $t$, $\rho(t,z)$ with an explicit scheme, because, this way, the equation is kept conservative at a discrete level. Additionally, the expression of the hiring rate $h([\rho])$ does not depend on the time step $\delta t $. However, the Courant-Friedrichs-Levi condition (\cite{bouchut2004nonlinear,leveque1992numerical}) is more restrictive: $1- \max_j(\mu_j) \delta t -\frac{\delta t}{\delta z} \ge 0$. The first order discretization is 
      $$
      \overbrace{\frac{\rho_j^{k+1}-\rho_j^k}{\delta t}+\frac{\rho_j^k-\rho_{j-1}^k}{\delta z}}^{Workforce \, evolution}=-\overbrace{\mu_j\rho_j^{k}}^{Attrition}+\overbrace{h_k\gamma_j}^{Hiring},
     \hbox{ and so } \
      \rho_j^{k+1}=\rho_{j}^k(1-\mu_j \delta t)+\delta t\left(h_k\gamma_j-\frac{\rho_{j}^k-\rho_{j-1}^k}{\delta z}\right).  
      $$
       Since we know that 
       $$
       \sum_{j=1}^{J}\omega_j\rho_j^k=\sum_{j=1}^{J}\omega_j\rho_j^{k+1},  \hbox{ we have, }   \      \omega_j\frac{\rho_j^{k+1}-\rho_j^k}{\delta t}+\omega_j\frac{\rho_j^k-\rho_{j-1}^k}{\delta z}=-\omega_j\mu_j\rho_j^{k}+\omega_jh_k\gamma_j,
        $$
         therefore, by summing from $j=1$ to $J$, one immediately gets 
       $$
       h_k=\frac{\sum_{j}^{}\left(\frac{\omega_j(\rho_j^k-\rho_{j-1}^k)}{\delta z}+\mu_j\rho_j^k\omega_j\right)}{\sum_{j}^{}\gamma_j \omega_j} \cdotp
       $$ 

 \section{Labor costs minimization}\label{ap4}
 
 We study here the problem of cost minimization (subsection \ref{sub32}).  { Let $E>0$ given and let $\rho_{eq}$ be a nonnegative stationary solution  of Equation \eqref{eq3} such that 
 \begin{equation}\label{contraint}
 E=\int_{z_{\rm min}}^{z_{\rm max}}{\rho_{eq}(z)zdz}.
 \end{equation}
  We define 
 $$C(\rho_{eq}) = \int_{z_{\rm min}}^{z_{\rm max}}{\rho_{eq}(z)w(z)dz}.$$
We consider the following minimisation problem:  can we find a nonnegative  stationary solution $\rho^*$ of Equation \ref{eq3}, with constraint \eqref{contraint}  such that
\begin{equation}\label{minim} 
C(\rho^*) \leq C(\rho_{eq}) \hbox{ ? }
\end{equation}
Where $\rho_{eq}$ is  a nonnegative stationary solution  of Equation \eqref{eq3} with \eqref{contraint}. We introduce the three following functions
$$f(z)=\int_{z}^{z_{\rm max}}{w(y)e^{-M(y)}dy}, \quad g(z)=\int_{z}^{z_{\rm max}}{ye^{-M(y)}dy}, \quad   d(z)=\frac{f(z)}{g(z)} \cdotp$$ The following Proposition holds }
        \begin{proposition} 
        {A solution of the minimization problem \eqref{minim}  is given by 
  $$\rho^*(z)=e^{-M(z)}b\mathbf{1}_{z\geq z_0}, \hbox{ with }
   b=\frac{E}{\int_{z_0}^{z_{\rm max}}{ze^{-M(z)}dz}},  \hbox{ and $z_0$ such that } d(z_0)=\min_z\left(d(z)\right). $$
   We then have $C=Ed(z_0).$  } 
        \end{proposition}
   \begin{proof}
   {
 Let $\rho_{eq}$ be a stationary state of Equation \eqref{eq3} such that constraint \eqref{contraint} is satisfied. We have 
 $$
 C=\int_{z_{\rm min}}^{z_{\rm max}}{w(z)\rho_{eq}(z)dz}=\int_{z_{\rm min}}^{z_{\rm max}}{w(z)e^{-M(z)}\rho_{eq}(z)e^{M(z)}dz}.
 $$
  We denote $Q(z)=\frac{\mathrm{d}(\rho_{eq}e^{M})}{\mathrm{d}z}(z)$, then, integrating by parts $$
  C=\int_{z_{\rm min}}^{z_{\rm max}}{\left(\int_{z}^{z_{\rm max}}{w(u)e^{-M(u)}du}\right)Q(z)dz}+{\left[\left(\int_{z_{\rm max}}^{z}{w(u)e^{-M(u)}du}\right)\rho^*(z)e^{M(z)}\right]_{z=z_{\rm min}}^{z=z_{\rm max}}}, 
  $$
   the last term vanishes thanks to the boundary condition $\rho_{eq}(z_{\rm min})=0$. Therefore, we obtain
 $$
 C=\int_{z_{\rm min}}^{z_{\rm max}}{\left(\int_{z}^{z_{\rm max}}{w(y)e^{-M(y)}dy}\right)Q(z)dz}=\int_{z_{\rm min}}^{z_{\rm max}}{f(z)Q(z)dz},
 $$
 and in the same way, we find 
 $$ E=\int_{z_{\rm min}}^{z_{\rm max}}{\left(\int_{z}^{z_{\rm max}}{ye^{-M(y)}dy}\right)Q(z)dz}=\int_{z_{\rm min}}^{z_{\rm max}}{g(z)Q(z)dz}, $$
 Consequently we obtain: 
 $$C=\int_{z_{\rm min}}^{z_{\rm max}}{\frac{f(z)}{g(z)}g(z)Q(z)dz}\geq E  \min_z\left(d(z)\right).
 $$
  By continuity, the  minimum of $d$ is reached at least on a point $z_0$. We then chose 
   $$Q(z)=b\delta_{z_0}=\frac{\mathrm{d}(\rho^*e^{M})}{\mathrm{d}z}(z), \hbox{ and so } \rho^*(z)=e^{-M(z)}b\mathbf{1}_{z\geq z_0},$$ where $b\geq 0$ is the positive constant such that the constraint \eqref{contraint} is satisfied. }This gives us the ideal age structure $\rho^*$ at the equilibrium state, and then we can deduce the hiring rate and profile 
$$
\gamma^*(z) =  \rho^{*'}+ \mu \rho^*=b\delta_{z_0}e^{-M(z)}.
$$
  \hfill \qed 
  \end{proof}  

\noindent {\bf Acknowledgments.} 
Marie Doumic was supported by the ERC Starting Grant SKIPPER$^{AD}$ (number 306321). 
Benoit Perthame and Delphine Salort were supported by the French "ANR blanche" project Kibord: ANR-13-BS01-0004 funded by the French Ministry of Research. 
\begin{scriptsize}
\bibliographystyle{plain}
\bibliography{bibli}
\end{scriptsize}

\end{document}